\documentclass[10pt,journal,a4paper]{IEEEtran} 

\usepackage{amssymb}
\usepackage{amsmath}
\usepackage{amsmath,bm}
\usepackage{amsthm}
\usepackage{graphicx}
\usepackage{cuted}
\usepackage{subfigure}
\usepackage{multirow}
\usepackage{makecell} 
\usepackage{cite}
\usepackage{enumerate}
\usepackage{enumitem}
\usepackage{stfloats}
\usepackage{algorithm}
\usepackage{algorithmic}  
\usepackage{color, soul}
\usepackage[labelsep=period]{caption}
\usepackage{hyperref}

\allowdisplaybreaks[4]

\DeclareMathOperator{\Tr}{Tr}
\DeclareMathOperator{\diag}{diag}
\DeclareMathOperator{\permanent}{Per}

\begin{document}
	\newtheorem{theorem}{\bf~~Theorem}
	\newtheorem{remark}{\bf~~Remark}
	\newtheorem{observation}{\bf~~Observation}
	\newtheorem{definition}{\bf~~Definition}
	\newtheorem{lemma}{\bf~~Lemma}
	\newtheorem{preliminary}{\bf~~Preliminary}
	\newtheorem{proposition}{\bf~~Proposition}
	\newtheorem{comment}{\bf~~Comment}
	\renewcommand\arraystretch{0.9}
	\title{{Revisiting Near-Far Field Boundary in Dual-Polarized XL-MIMO Systems}}
	\author{\normalsize \IEEEauthorblockN{
			{Shuhao Zeng}, \IEEEmembership{\normalsize Member, IEEE},
			{Boya Di}, \IEEEmembership{\normalsize Member, IEEE},
			{Hongliang Zhang}, \IEEEmembership{\normalsize Member, IEEE},\\
			{Zhu Han}, \IEEEmembership{\normalsize Fellow, IEEE},
			{and H. Vincent Poor}, \IEEEmembership{\normalsize Fellow, IEEE}}
		\thanks{Shuhao Zeng is with School of Electronic and Computer Engineering, Peking University Shenzhen Graduate School, Shenzhen, China, and also with Department of Electrical and Computer Engineering, Princeton University, NJ, USA (email: shuhao.zeng96@gmail.com).}
		\thanks{Boya Di and Hongliang Zhang are with School of Electronics, Peking University, Beijing 100871, China (email: hongliang.zhang@pku.edu.cn; boya.di@pku.edu.cn).}
		\thanks{Zhu Han is with Electrical and Computer Engineering Department, University of Houston, Houston, TX, USA, and also with the Department of Computer Science and Engineering, Kyung Hee University, Seoul, South Korea (email: zhan2@uh.edu).}		
		\thanks{H. Vincent Poor is with Department of Electrical and Computer Engineering, Princeton University, NJ, USA (email: poor@princeton.edu).}
	}
	
	\maketitle
	\begin{abstract}
		Extremely large-scale multiple-input multiple-output~(XL-MIMO) is expected to be an important technology in future sixth generation~(6G) networks. Compared with conventional single-polarized XL-MIMO, where signals are transmitted and received in only one polarization direction, dual-polarized XL-MIMO systems achieve higher data rate by improving multiplexing performances, and thus are the focus of this paper. Due to enlarged aperture, near-field regions become non-negligible in XL-MIMO communications, necessitating accurate near-far field boundary characterizations. However, existing boundaries developed for single-polarized systems only consider phase or power differences across array elements while irrespective of cross-polarization discrimination~(XPD) variances in dual-polarized XL-MIMO systems, deteriorating transmit covariance optimization performances. In this paper, we revisit near-far field boundaries for dual-polarized XL-MIMO systems by taking XPD differences into account, which faces the following challenge. Unlike existing near-far field boundaries, which only need to consider co-polarized channel components, deriving boundaries for dual-polarized XL-MIMO systems requires modeling joint effects of co-polarized and cross-polarized components. To address this issue, we model XPD variations across antennas and introduce a non-uniform XPD distance to complement existing near-far field boundaries. Based on the new distance criterion, we propose an efficient scheme to optimize transmit covariance. Numerical results validate our analysis and demonstrate the proposed algorithm's effectiveness.
	\end{abstract}

\vspace{-.4cm}
		
	\section{Introduction}
	Extremely large-scale multiple-input multiple-output~(XL-MIMO) is expected to play an important role in the forthcoming sixth generation~(6G) networks~\cite{Wu_distance_aware_2022,Deng_RHS_2022}, where numerous antenna elements are integrated into a compact space. Benefiting from the large radiation aperture of the resulting antenna array, XL-MIMO can achieve high directive gain~\cite{Han_XLMIMO_2020},~\cite{Zeng_multi_user_RRS}, and {therefore is capable of supporting high-speed data transmissions~\cite{Guo_2023}}. Most existing works on this area mainly focus on single-polarized XL-MIMO, {where both the base station~(BS) and the user equipment~(UE) are equipped with single-polarized antennas~\cite{Zeng_RRS_letter_2022,Guo_2024}}. Wireless signals are essentially a kind of electromagnetic~(EM) waves with polarization characteristics, i.e., the orientation of electric field strength varies with time~\cite{Guru_EM_field_2004}. Since the electric field strength can only oscillate in the plane perpendicular to the propagation direction of the EM waves, there are two orthogonal polarizations~\cite{Stutzman_polarization_2018}, i.e., horizontal and vertical polarizations. However, in single-polarized XL-MIMO, transceivers can only transmit and receive signals in one polarization direction~\cite{Zeng_both_2021}, where the polarization dimension is not fully exploited, thereby leading to degraded communication performance.
	
	By deploying dual-polarized antennas at both the BS and the UEs, dual-polarized XL-MIMO serves as a promising solution to the above issue. Since each dual-polarized antenna element consists of two co-located single-polarized elements, as shown in Fig.~\ref{sysmodel}, the number of antenna elements can be doubled in the same physical enclosure~\cite{Kim_limited_feedback}. Further, benefiting from capability of transmitting and receiving signals in both polarizations, dual-polarized XL-MIMO systems can achieve better multiplexing performance than conventional single-polarized ones~\cite{Zeng_dual_polarized_RIS}.


{
	It is crucial to accurately characterize near-far field boundaries for XL-MIMO communication systems. Specifically, due to the increased aperture of the antenna array, the Fresnel region (i.e., radiating near-field region of the antenna array) is significantly enlarged~\cite{Yue_channel_2024}. Therefore, users are more likely to be situated in the near-field region. Note that channel modeling in the near-field and far-field differs significantly~\cite{Sun_differentiate}. Therefore, it is critical to investigate the boundary between the near-field and far-field, as it affects both channel characterizations and communication performance optimizations. Existing works on near-far field boundary mainly focus on single-polarized XL-MIMO systems~\cite{Selvan_Rayleigh_2017,Lu_XL_MIMO_ICC,Lu_XL_MIMO_2022,effective_Rayleigh}. The classic near-field boundary is called Rayleigh distance~\cite{Selvan_Rayleigh_2017}, which is defined from the perspective of phase error. When the distance between a user and the antenna array exceeds the Rayleigh distance, the maximum phase error across the antenna aperture between using the planar-wave model and the spherical-wave model is no more than $\frac{\pi}{8}$. However, when maximal ratio combining is applied at the antenna array, signal phases can be perfectly aligned, rendering the received power dependent solely on the channel's amplitude response. Therefore, the authors in~\cite{Lu_XL_MIMO_ICC} propose an alternative critical distance as the near-field boundary, defined as the power ratio between the weakest and strongest array elements remaining above a specified threshold. The authors in~\cite{Lu_XL_MIMO_2022} further extend this critical distance to the uniform planar array case by considering the projected aperture across the antenna array in the derivation. Note that these distances do not directly take the communication performances into account. Therefore, in~\cite{effective_Rayleigh}, the authors propose effective Rayleigh distances from a beamforming perspective. This distance ensures that the normalized beamforming gain under the far-field assumption remains no less than $95\%$.}
	
	{However, none of the above near-far field boundaries have considered the variance of cross-polarization discrimination~(XPD) across different antenna elements in dual-polarized XL-MIMO systems, which can deteriorate transmit covariance optimization performances.} Specifically, in dual-polarized wireless communications, XPD is a commonly used performance metric to characterize the channel's ability to preserve radiated polarization purity between horizontally and vertically polarized signals\footnote{Formally, XPD can be defined as the ratio between the co-polarized and cross-polarized channel power gains, which can be found in (\ref{XPD_RRS_2_UE}).}. Due to the enlarged antenna aperture in dual-polarized XL-MIMO systems, the XPDs of channels vary across different BS antennas\footnote{Due to the increased aperture of the antenna array~\cite{Zeng_IMT2030}, the propagation distances from the BS to the UE and the angles of departure~(AoDs) of the BS relative to environmental scatters vary across BS antennas. Therefore, the XPD differ across antennas. This will be further elaborated on in Section~\ref{sec_performance_analysis}.}. Note that most existing near-far field boundaries only consider phase or power differences across array elements while irrespective of such XPD variations. Therefore, they cannot be used to accurately characterize dual-polarized XL-MIMO channels, leading to deteriorated transmit covariance optimization performances\footnote{Due to the variant XPD over antenna elements in dual-polarized XL-MIMO systems, existing transmit covariance optimization schemes designed for dual-polarized massive MIMO systems cannot be directly applied, since they are based on the assumption of equal XPD and pathloss across array elements.}. 
	
	{In this paper, we revisit near-far field boundaries for dual-polarized XL-MIMO systems by taking XPD differences across BS antennas into consideration. Compared with existing power-variation-based near-field boundary developed for single-polarized systems~\cite{Lu_XL_MIMO_ICC,Lu_XL_MIMO_2022}, incorporating XPD variations brings two new challenges. First, unlike the derivation of power-variation-based distances, which only considers propagation distances, the derivation of non-uniform XPD distances should account for both propagation distances and angle-of-departure~(AoD). Second, when deriving non-uniform XPD distances, it is necessary to consider both co-polarized and cross-polarized channel components. In contrast, the existing power-variation-based distance focuses only on the effect of co-polarized components, as it is primarily designed for single-polarized systems.}
	


	To cope with the above challenges, we first model dual-polarized XL-MIMO channels, accounting for XPD differences across BS antennas. Based on this channel model, we introduce a non-uniform XPD distance to complement existing near-far field boundary. When the distance between the BS and the user is shorter than the non-uniform XPD distance, the XPD and pathloss vary across different BS antennas, and consequently existing transmit covariance matrix design methods cannot be directly applied. Therefore, we formulate a transmit covariance optimization problem to maximize the ergodic capacity, where an efficient transmit covariance matrix optimization algorithm is then proposed to solve the formulated problem. Numerical results verify our analysis and demonstrate the effectiveness of the proposed algorithm. The main contributions of this paper are summarized below,
	\begin{itemize}
		\item We propose a dual-polarized XL-MIMO channel model, highlighting the discrepancies in XPD across BS antenna elements by jointly considering the influences of propagation distances and AoD. Based on such models, we introduce a new distance criterion, termed the \textit{non-uniform XPD distance}, which accounts for the XPD variations across array elements and complements existing near-far field boundaries.
		
		\item We formulate a transmit covariance matrix optimization problem to maximize the ergodic capacity\footnote{The optimal transmit covariance can be achieved through hybrid beamforming structure at the BS, such as that in~\cite{Zheng_joint_2024} and~\cite{Shi_spatial_2024}. Further, an optimized transmit covariance matrix can be utilized to guide the design of beamformers at the BS.}, which is decomposed into two sub-problems: an eigenvector design subproblem and a power allocation subproblem. Based on the proposed non-uniform XPD distance, we design an iterative transmit covariance optimization algorithm, where a closed-form solution is first derived for the eigenvector design problem, and then an efficient scheme accounting for non-uniform XPD is proposed to solve the power allocation subproblem.
		\item Simulation results verify the analytical results and show the effectiveness of the proposed algorithm. Further, insights are gained regarding the influences of pathloss and XPD on transmit covariance matrix design. {Specifically, the transmit power (i.e., the diagonal elements of the transmit covariance matrix) tends to be allocated to BS antennas with less severe pathloss or higher XPD.}
	\end{itemize}
	
	The rest of this paper is organized as follows. In Section~\ref{sec_system_model}, a dual-polarized XL-MIMO system is modeled, based on which we introduce non-uniform XPD distances to complement existing near-far field boundary in Section~\ref{sec_performance_analysis}. In Section~\ref{sec_problem_formulation_and_decomposition}, we formulate a transmit covariance optimization problem and decompose it into two subproblems. A permanent-based algorithm is proposed to solve the transmit covariance optimization problem in Section~\ref{sec_algorithm_design}. In Section~\ref{sec_simulation}, simulation results are presented, and finally, conclusions are drawn in Section~\ref{sec_conclusion}.

	\section{System Model}
	\label{sec_system_model}
	\begin{figure}[!t]
		\centering
		\includegraphics[width=0.41\textwidth]{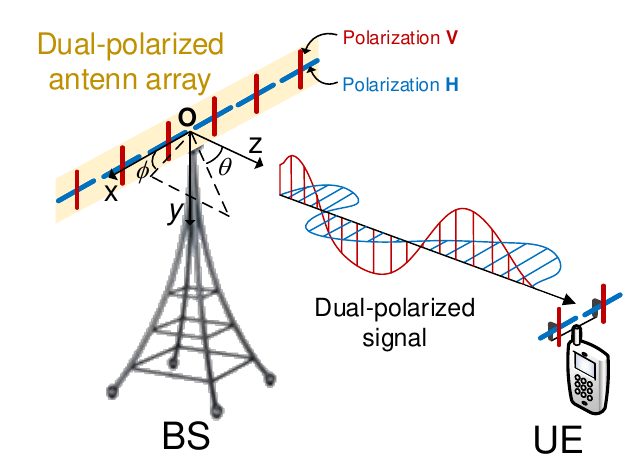}
		\caption{System model of a downlink dual-polarized XL-MIMO network.}
				\vspace{-3mm}
		\label{sysmodel}
	\end{figure}
	\subsection{Scenario Description}
	\label{scenario_description}
	In this paper, we consider a narrow-band downlink XL-MIMO system, where a BS equipped with an extremely large antenna array\footnote{The antenna array at the BS can be either linear or planar.} transmits to one multi-antenna user. Each BS and UE antenna is dual-polarized, consisting of one vertical- and one horizontal-polarized component that are co-located~\cite{Ozdogan_dual_polarized_2023}. Further, denote the number of antennas at the BS and the UE by $M$ and $N$, respectively. 
	
	For ease of exposition, we introduce Cartesian coordinates, where the $x-y$ plane coincides with the antenna array at the BS and the $z$-axis is vertical to the antenna array, as shown in Fig.~\ref{sysmodel}. Under the introduced coordinate system, we use $r_U$, $\theta_U$, and $\phi_U$ to represent the distance between the UE and the origin, the zenith angle of the UE, and the azimuth angle of the UE, respectively. Therefore, the coordinates of the UE can be expressed as $\bm{q}_U=(r_U\sin\theta_U\cos\phi_U,r_U\sin\theta_U\sin\phi_U,r_U\cos\theta_U)^{\mathrm{T}}$.
	\subsection{Channel Model}
	\label{subsec_channel_model}
	Since both the BS antennas and UE antennas are dual-polarized, the channel matrix $\bm{G}\in \mathbb{C}^{2N\times 2M}$ from the BS to the UE consists of four components, i.e.,
	\begin{equation}
		\label{channel}
		\bm{G}=\left[
		\begin{matrix}
			\bm{G}^{(VV)} & \bm{G}^{(VH)}  \\
			\bm{G}^{(HV)} & \bm{G}^{(HH)}
		\end{matrix}
		\right],
	\end{equation}
	where, $\forall i,j\in\{V,H\}$, the $(n,m)$-th element of the matrix $\bm{G}^{(ji)}\in\mathbb{C}^{N\times M}$, i.e., $g^{(ji)}_{n,m}$, is defined as the channel from the $m$-th BS antenna in polarization $i$ to the $n$-th UE antenna in polarization $j$. To this end, the propagation matrix $\bm{G}$ describes the relation from $V$ to $V$, $H$ to $V$, $V$ to $H$, and $H$ to $H$ polarized waves. 
	 
	Due to  interactions between EM waves and environmental scatterers, propagation from the BS to the UE can change polarizations. To characterize the channel's ability to preserve radiated polarization purity between horizontally and vertically polarized signals, XPD is commonly used as a performance metric~\cite{bruno_MIMO_XPC_2013}. Formally, the XPD for the channel from the $m$-th BS antenna to the $n$-th UE antenna is defined as\footnote{Equation~(\ref{XPD_RRS_2_UE}) serves as the formal definition for XPD, while (\ref{def_l}) is the definition for the normalized parameter $l_m$ and (\ref{XPD_model}) is the proposed XPD model. {Further, the value of XPD is independent of the pathloss $\beta_m$, since the pathloss of co-polarized channel components and that of cross-polarized channel components are assumed to be equal~\cite{Ozdogan_dual_polarized_2023}.}}~\cite{Ozdogan_dual_polarized_2023}
	\begin{align}
		\label{XPD_RRS_2_UE}
		XPD_{n,m}\triangleq\frac{\mathbb{E}\{|g^{(HH)}_{n,m}|^2\}}{\mathbb{E}\{|g^{(VH)}_{n,m}|^2\}}=\frac{\mathbb{E}\{|g^{(VV)}_{n,m}|^2\}}{\mathbb{E}\{|g^{(HV)}_{n,m}|^2\}}.
	\end{align}
	High channel XPD indicates that the received co-polarized signal at the UE is significantly stronger than the cross-polarized signal, and thus the channels can maintain radiated polarization purity. Conversely, when channel XPD is low, most transmit power is converted to the orthogonal polarization. 
	
	Due to the increased aperture of the BS antenna array, the XPD of channels varies for different BS antennas. This is because distinct BS antennas can observe diverse propagation environments, including propagation distances and environment scatterer distributions, which will be elaborated on in Section~\ref{sec_performance_analysis}. On the contrary, since the number of UE antennas is much smaller than that of BS antennas, i.e., $N\ll M$, the XPD of channels for different UE antennas is regarded as the same. Therefore, in the following, we neglect the index $n$ in the expression of XPD for brevity, i.e.,
	\begin{align}
		XPD_{n,m}=XPD_{m},~\forall n.
	\end{align}

	{Define $l_m\in [0,1]$ as
	\begin{align}
		\label{def_l}
		l_m=\frac{1}{1+XPD_{m}}.
	\end{align}
	By mapping $XPD_m$, which takes values in $[0,\infty)$, to the interval $[0,1]$, $l_m$ provides a normalized representation that maintains a one-to-one correspondence with $XPD_m$, facilitating easier comparisons and analyses within a bounded range.} { Motivated by the definitions of XPD and $l_m$ in (\ref{XPD_RRS_2_UE}) and (\ref{def_l}), we can model the influence of XPD on channel power gain as follows,}
	\begin{align}
		\label{pathloss}
		\mathbb{E}(|g^{(HH)}_{n,m}|^2)&=\mathbb{E}(|g^{(VV)}_{n,m}|^2)=\beta_m(1-l_{m}),\\
		\label{pathloss_cross}
		\mathbb{E}(|g^{(VH)}_{n,m}|^2)&=\mathbb{E}(|g^{(HV)}_{n,m}|^2)=\beta_ml_{m},
	\end{align}
	where $\beta_m$ is pathloss. Here, the index $m$ indicates that the variation of pathloss over BS antennas cannot be neglected due to the enlarged aperture of the BS antenna array~\cite{Zeng_multi_user_RRS}. Further, we omit the subscript $n$ representing the UE antenna elements because the aperture of the UE antenna array is relatively small, making the pathloss from different UE antenna elements to the BS approximately the same.


	{
	Next, we discuss the correlation structures of the dual-polarized XL-MIMO channels. We assume non-line-of-sight~(NLoS) communications between the BS and the UE, and thus each channel is assumed to possess zero mean~\cite{Yuanwei_tutorial_2023}, i.e., $\mathbb{E}\{g^{(ji)}_{n,m}\}=0$.
	We first focus on the correlation between \emph{different polarizations}. Various measurements indicate that in NLoS scenarios, the transmit and receive cross-polarization correlations are nearly zero~\cite{bruno_MIMO_XPC_2013}, i.e., for any $(n_1,m_1)$ and $(n_2,m_2)$, we can assume
	\begin{align}
		\mathbb{E}\{g^{(VV)}_{n_1,m_1}({g}^{(VH)}_{n_2,m_2})^*\}=\mathbb{E}\{{g}^{(HV)}_{n_1,m_1}({g}^{(HH)}_{n_2,m_2})^*\}=0, 
	\end{align}
	and
	\begin{align}
		\mathbb{E}\{{g}^{(VV)}_{n_1,m_1}({g}^{(HV)}_{n_2,m_2})^*\}=\mathbb{E}\{{g}^{(HH)}_{n_1,m_1}({g}^{(VH)}_{n_2,m_2})^*\}=0, 
	\end{align}
	respectively. In addition, both the co-polarization correlation and the anti-polarization correlation are found to be nearly zero as well~\cite{yu_dual_polarized_2022}, i.e., for any $(n_1,m_1)$ and $(n_2,m_2)$, we can assume
	\begin{align}
		\mathbb{E}\{{g}^{(VV)}_{n_1,m_1}({g}^{(HH)}_{n_2,m_2})^*\}=0,\\
		\mathbb{E}\{{g}^{(HV)}_{n_1,m_1}({g}^{(VH)}_{n_1,m_1})^*\}=0.
	\end{align}
	Further, we assume that both the BS antenna elements and the UE antenna elements are well spaced so that the MIMO channels are spatially uncorrelated\footnote{In this paper, we aim to model the variance of XPD with antenna elements in dual-polarized XL-MIMO networks and investigate its influence on transmit covariance optimizations. Therefore, we neglect the spatial correlation due to near-field spherical-wave propagation~\cite{Dong_spatial_correlation_early} and assume i.i.d. channels for simplicity so as to gain insights into these open problems. The investigation of transmit covariance design under spatially correlated XL-MIMO channels is an interesting topic for future work.}, i.e., $\forall i, j\in \{V,H\}$, and for any $(n_1,m_1)$ and $(n_2,m_2)$, we have
	\begin{align}
		\mathbb{E}\{{g}^{(ji)}_{n_1,m_1}({g}^{(ji)}_{n_2,m_2})^*\}=0.
	\end{align}}
	
	Define $\bm{x}^{(V)}\in\mathbb{C}^{M\times 1}$ and $\bm{x}^{(H)}\in\mathbb{C}^{M\times 1}$ as the transmitted signals of the $V$-polarized and $H$-polarized antennas at the BS, respectively. Further, by cascading $\bm{x}^{(V)}$ and $\bm{x}^{(H)}$, we can acquire the overall transmitted signals $\bm{x}$, i.e., $\bm{x}=[(\bm{x}^{(V)})^{\mathrm{T}},(\bm{x}^{(H)})^{\mathrm{T}}]^{\mathrm{T}}$. Based on the above channel model, the received signal $\bm{y}\in\mathbb{C}^{2N\times 1}$ at the UE can be written as,
	\begin{align}
		\label{rec_sig}
		\bm{y}=\bm{G}\bm{x}+\bm{n},
	\end{align}
	where $\bm{n}\in \mathbb{C}^{2N\times 1}$ is additive white Gaussian noise~(AWGN) at the UE with zero mean and $\sigma^2$ as variance. Based on the received signal model in (\ref{rec_sig}), the ergodic capacity of dual-polarized XL-MIMO systems is given by\footnote{We assume that the statistical channel state information~(CSI) is available to the transmitter while the instantaneous CSI is unknown to the transmitter.}~\cite{Shin_Keyhole_2003,A_2005},
	\begin{align}
		\label{ergodic_capacity}
		C=\mathbb{E}\{\log_2\det(\bm{I}_{2N}+\gamma \bm{G}\bm{Q}\bm{G}^{\dagger})\}.
	\end{align}
	Here, $\mathbb{E}\{\cdot\}$ denotes expectation, $\gamma=\frac{P}{\sigma^2}$ is transmit signal-to-noise ratio~(SNR), where $P$ represents the transmit power, and $\bm{Q}$ is normalized transmit covariance matrix, i.e., $\bm{Q}=\frac{1}{P}\mathbb{E}\{\bm{x}\bm{x}^{\dagger}\}$. 
	{
	\begin{remark}
		XPDs have an impact on system capacity by influencing the distribution of channel matrix $\bm{G}$. This is because according to (\ref{ergodic_capacity}), the ergodic capacity of the considered XL-MIMO system is dependent on the distribution of channel matrix  $\bm{G}$. Further, based on (\ref{def_l})-(\ref{pathloss_cross}), it can be found that XPDs have an impact on the distribution of channel matrix $\bm{G}$.
	\end{remark}}
	
	\section{Near-Far Field Boundary Analysis}
	\label{sec_performance_analysis}
	In this section, we first characterize the XPD of dual-polarized XL-MIMO channels. On this basis, we introduce \textit{non-uniform XPD distances} to complement existing near-far field boundary. Then, extensive discussions on performance analysis are provided, including non-uniform XPD aperture, and the challenge brought by the non-uniform XPD and pathloss to transmit covariance matrix design.
	\vspace{-0.2cm}
	\subsection{Characterizations of XPD}
	There are two main factors that can result in the discrepancies in XPDs across BS antenna elements:
	\begin{itemize}
		\item \emph{Different propagation distances}: Both co-polarized and cross-polarized signal components decrease with the propagation distance $d_m$ from the $m$-th BS antenna to the UE, where the attenuation with distance is best modeled as exponential decay~\cite{Shafi_polarized_2006}. Further, decay exponents for the co-polarized and cross-polarized signals are different. Therefore, XPD, representing the power ratio between received co-polarized and cross-polarized signals, also varies with distance. Due to the increased aperture of the BS antenna array, the discrepancies in distances from different BS antenna elements to the UE cannot be neglected, leading to the variation of XPD with different BS antenna elements. 
		\item \emph{Different angles of departure~(AoD)}: Due to the differing AoDs of various BS antennas relative to clusters in wireless propagation environments, signals transmitted by different BS antennas exhibit distinct polarization directions upon incidence on the clusters. Consequently, this variation leads to differences in the polarization directions of scattered signals and their corresponding received signals, resulting in distinct XPD values.
		
	\end{itemize} 
	
	To account for the influence of propagation distances $d_{m}$, the following term is incorporated into the expression of XPD~\cite{Shafi_polarized_2006}, i.e.,
	\begin{align}
		\label{influence_distance}
		\chi_1(d_{m})=XPD_{d_{m}=1}\times(d_{m})^{{\eta}}.
	\end{align}
	Here, $XPD_{d_{m}=1}$ represents the XPD under unit propagation distance. Further, $\eta$ is a positive constant, which indicates that the XPD increases with propagation distances. This is because the decay of the co-polarized component with distance is less severe than the cross-polarized one~\cite{Shafi_polarized_2006}. 
	
	Denote the influence of AoDs on XPD by $\chi_2$. {Based on~\cite{Quitin_multipolarized_2009}, the expression of $\chi_2$ is given by (\ref{influence_AoD}) shown at the bottom of the next page.}
	\begin{figure*}[!hb]
		\rule[-12pt]{17.5cm}{0.05em}
		\begin{equation}
			\setlength{\abovedisplayskip}{3pt}
			\setlength{\belowdisplayskip}{-3pt}
			\begin{aligned}
				\label{influence_AoD}
				&\chi_2(\{\phi_m^{(l)}\}_l)\\
				&=\chi^{(0)}_2\frac{\sum\limits_{l=1}^L 2e^{\frac{\Delta \varphi^{(l)}}{\sqrt{2}\omega^{(l)}}}(1+2(\omega^{(l)})^2-\cos(2\phi_m^{(l)}))-2(1+2(\omega^{(l)})^2)+2\cos(\Delta \varphi^{(l)})\cos(2\phi_m^{(l)})-2\sqrt{2}\omega^{(l)}\sin(\Delta \varphi^{(l)})\cos(2\phi_m^{(l)})}{\sum\limits_{l=1}^L 2e^{\frac{\Delta \varphi^{(l)}}{\sqrt{2}\omega^{(l)}}}(1+2(\omega^{(l)})^2+\cos(2\phi_m^{(l)}))-2(1+2(\omega^{(l)})^2)-2\cos(\Delta \varphi^{(l)})\cos(2\phi_m^{(l)})+2\sqrt{2}\omega^{(l)}\sin(\Delta \varphi^{(l)})\cos(2\phi_m^{(l)})},
			\end{aligned}
		\end{equation}
	\end{figure*}
	{Here, the numerator models the dependence of the co-polarized channel power gain, i.e., $\mathbb{E}\{|g^{(ii)}_{n,m}|^2\}$~$(i\in\{V, H\})$, over the angles of departure while the denominator models the relationship between the cross-polarized channel power gain, i.e., $\mathbb{E}\{|g^{(ji)}_{n,m}|^2\}$~$(i\neq j)$, and the angles of departure.} Further, $L$ represents the number of clusters in wireless propagation environments, {$\phi_m^{(l)}$ is the mean azimuth departure angle of the $l$-th cluster corresponding to the $m$-th BS antenna}\footnote{For simplicity, the effects of elevation angles are neglected here~\cite{Quitin_multipolarized_2009}}, which varies with different BS antennas. {Further, the truncation spread $\Delta \varphi^{(l)}$ describes the range of the azimuth departure angle at the BS, i.e., all the possible azimuth departure angles for the multipath signals within the $l$-th cluster belong to the range $[\phi^{(l)}_m-\Delta \varphi^{(l)}/2,\phi^{(l)}_m+\Delta \varphi^{(l)}/2]$.} In addition, $\omega^{(l)}$ is the azimuth spread corresponding to the $l$-th cluster, and $\chi_2^{(0)}$ is a normalized factor which is selected such that the term $\chi_2$ is equal to $1$ for the $0$-th BS antenna (the one at the center of the array), i.e., $\chi_2(\{\phi_0^{(l)}\}_l)=1$. 
	
	By combining (\ref{influence_distance}) and (\ref{influence_AoD}), we propose an XPD model, as shown in the following theorem.
	\begin{theorem}
		\label{theorem_XPD_model}
		The XPD of the channel from the $m$-th BS antenna to the UE can be modeled as
		\begin{align}
			\label{XPD_model}
			XPD_m=\chi_1(d_{m})\chi_2(\{\phi_m^{(l)}\}_l),
		\end{align}
		where $\chi_1(d_{m})$ defined in (\ref{influence_distance}) and $\chi_2(\{\phi_m^{(l)}\}_l)$ defined in (\ref{influence_AoD}) characterize the influences of discrepancies in propagation distances and AoDs across different BS antennas, respectively\footnote{Here, we do not assume that the propagation distances and AoDs are independent. Actually, as the propagation distance varies, the user's location changes accordingly. This shift in position can result in different propagation paths, ultimately leading to variations in AoDs.}.
	\end{theorem}
	 
\subsection{Non-Uniform XPD Distance}
Based on the XPD model in Theorem~\ref{theorem_XPD_model}, we introduce a non-uniform XPD distance as a complement to existing Rayleigh distances for separating near-field and far-field regions. The non-uniform XPD distance $r_U^{th}$ represents a threshold for the BS-UE distance $r_U$, where the discrepancy in XPD across different BS antennas cannot be neglected when $r_U<r_U^{th}$. The formal definition of $r_U^{th}$ is given in the following.

{
\begin{definition}
	\label{def_different_XPD_distance}
	Given the aperture of the BS antenna array, the non-uniform XPD distance $r_U^{th}$ is defined as\footnote{We do not use the definition of $\frac{\max_m (\chi_1 \chi_2)}{\min_m (\chi_1 \chi_2)} \geq r_u^{th}$, as it inherently couples the propagation-dependent term \(\chi_1\) and the AoD-dependent term \(\chi_2\). This coupling complicates the analytical derivation of closed-form expressions for non-uniform XPD distances. Therefore, to gain insights into the non-uniform XPD distances and non-uniform apertures, we decouple \(\chi_1\) and \(\chi_2\) by separately considering $\frac{\max_m  \chi_1}{\min_m \chi_1} \geq r_1^{th}$ and $\frac{\max_m  \chi_2}{\min_m \chi_2} \geq r_2^{th}$ in this paper. Moreover, it can be found the non-uniform XPD distances under the two definitions are similar, which further justifies the proposed definition. The same rationale applies to the definition of non-uniform XPD aperture in (\ref{def_non_unifrom_XPD_aperture}).}
	\begin{align}
		\label{def_non_unifrom_XPD_distance}
		r_U^{th}\triangleq
		\operatorname*{argmax}\limits_{r_U}  \left(\frac{\max_m\chi_1}{\min_m\chi_1}\ge\gamma^{th}_1\right)\text{and}  \left(\frac{\max_m\chi_2}{\min_m\chi_2}\ge\gamma^{th}_2\right),
	\end{align} 
	where $r_U$ represents the distance between the UE and the center of the BS antenna array, the definition of $\chi_1$ is given in (\ref{influence_distance}), and $m$ represents the index of the BS antennas.
	
	The definition in (\ref{def_non_unifrom_XPD_distance}) indicates that when the BS-UE distance $r_U\le r_U^{th}$, the ratio between the maximum and minimum values of the XPD component $\chi_1$ and that for $\chi_2$ are no smaller than predetermined thresholds.
\end{definition}}

Before providing an exact expression for the non-uniform XPD distance, we first introduce the definitions of some parameters. We use $k$ to represent the ratio between the size of the antenna array in the $y$-axis and that in the $x$-axis. Further, we define parameter $\delta_U$ as
\begin{align}
	\label{delta_U}
	\delta_U=\left|\frac{\sin\theta_U\cos\phi_U+k\sin\theta_U\sin\phi_U}{\sqrt{1+k^2}}\right|,
\end{align}
Based on Definition~\ref{def_different_XPD_distance} and the XPD model in Theorem~\ref{theorem_XPD_model}, we acquire the exact expression of $r_U^{th}$, as shown in the following theorem.

{
\begin{theorem}
	\label{theorem_maximum_user_distance}
	The non-uniform XPD distance $r_U^{th}$ is given by
	\begin{align}
		\label{expression_distance_th}
		r_U^{th}=&\max\Bigg\{r_1^{th},\notag\\
		&\frac{\gamma_2^{th}+1}{\gamma_2^{th}-1}\left|\sum_{l}\frac{\partial \chi_2}{\partial \phi_m^{(l)}} \bigg|_{\phi_m^{(l)}=\phi_0^{(l)}}\frac{\sin(\arctan(k)+\phi_0^{(l)})D}{2c^{(l)}\sin\theta_0^{(l)}}\right|\Bigg\}.
	\end{align}
	where $\eta$ is defined in (\ref{influence_distance}), $D$ represents the diagonal dimension of the BS antenna array, $\delta_U$ is defined in (\ref{delta_U}), $k$ represents the ratio between the size of the antenna array in the $y$-axis and that in the $x$-axis, $c^{(l)}$ is the ratio between the BS-scatterer distance and the BS-UE distance, i.e., $c^{(l)}=\frac{r_0^{(l)}}{r_U}$\footnote{Note that as the user moves away from the BS, the distance between the BS and scatterer clusters also becomes larger~\cite{Zhou_geo_2019}. Therefore, to characterize this trend, we model the distance $r_0^{(l)}$ between the $l$-th cluster and the BS by $r_0^{(l)}=c^{(l)}r_U$, where $r_U$ is the distance between the user and the BS.}, and $(\phi_m^{(l)},\theta_m^{(l)})$ are the mean azimuth and elevation departure angle of the $l$-th cluster corresponding to the $m$-th BS antenna, respectively. Further, $r_1^{th}$ characterizes the influence of the XPD term $\chi_1$ on the non-uniform XPD distance, and its expression can be found in (\ref{pattern_feed}) at the bottom of the next page. Here, $I(\cdot)$ represents the indicator function.

	\begin{figure*}[!hb]
		\rule[-12pt]{17.5cm}{0.05em}
		\begin{equation}
			\setlength{\abovedisplayskip}{3pt}
			\setlength{\belowdisplayskip}{-3pt}
			{
			\begin{aligned}
				\label{pattern_feed}
				r_1^{th}=\left(1-I\left(\delta_U- \sqrt{1-\frac{4}{(\gamma_1^{th})^{\frac{2}{\eta}}+3}}\right)\right)\frac{-D\delta_U-\sqrt{((\gamma_1^{th})^\frac{2}{\eta}-1)D^2(1-\delta_U^2)}}{2(1-(\gamma_1^{th})^\frac{2}{\eta}+(\gamma_1^{th})^\frac{2}{\eta}\delta_U^2)}+I\left(\delta_U- \sqrt{1-\frac{4}{(\gamma_1^{th})^{\frac{2}{\eta}}+3}}\right)\frac{\gamma_1^{th}+1}{\gamma_1^{th}-1}\frac{\eta D \delta_U}{2}.
			\end{aligned}}
		\end{equation}
	\end{figure*}
\end{theorem}}
\begin{proof}
	See Appendix~\ref{app_theorem_maximum_distance}.
\end{proof}

Based on Theorem~\ref{theorem_maximum_user_distance}, we can arrive at the following remark.
\begin{remark}
	Unlike the conventional near-far field boundary based on Rayleigh distance $r_{Ray}$, which is proportional to the square of the diagonal dimension of the BS antenna array, i.e., $r_{Ray}\propto D^2$, the non-uniform XPD distance is proportional to the diagonal dimension of the BS antenna array, i.e., $r_U^{th}\propto D$. 
\end{remark}
Note that the XPD factor $\eta$ given in (\ref{influence_distance}) is usually smaller than pathloss component~\cite{Shafi_polarized_2006}. Therefore, we have the following remark.
\begin{remark}
	\label{remark_shorter_XPD_distance_than_power_distance}
	When the BS-user distance is shorter than the non-uniform XPD distance, both the XPD and pathloss (i.e., $\beta_m$) vary across BS antenna elements.
	%
\end{remark} 
	\vspace{-0.2cm}
	\subsection{Extensive Discussions on System Performances}
	In the following, we try to answer another important question on non-uniform XPD, i.e., given the distance between the BS and the UE, how large should the antenna aperture be to achieve significant variations in XPD across different BS antennas, which is also termed as non-uniform XPD aperture\footnote{When analyzing non-uniform XPD aperture, the influences of both propagation distances and AoD are considered. Therefore, we cannot simply acquire the expression of non-uniform XPD aperture from Theorem~\ref{theorem_maximum_user_distance}.}.
	\subsubsection{Non-Uniform XPD Aperture}
	 The formal definition of the non-uniform XPD aperture $A^{th}$ is given in the following theorem.
	{\begin{definition}
		\label{def_different_XPD}
		Given the distance between the BS and the UE, the non-uniform XPD aperture $A^{th}$ is defined as
\begin{align}
	\label{def_non_unifrom_XPD_aperture}
	A^{th}\triangleq
	\operatorname*{argmin}\limits_{A}  \left(\frac{\max_m\chi_1}{\min_m\chi_1}\ge\gamma^{th}_1\right)\text{and}  \left(\frac{\max_m\chi_2}{\min_m\chi_2}\ge\gamma^{th}_2\right),
\end{align} 
	where $A$ represents the aperture of the BS antenna array, $m$ represents the index of BS antennas, $\chi_1$ and $\chi_2$ are XPD components defined in (\ref{influence_distance}) and (\ref{influence_AoD}), respectively, and $\gamma^{th}_1$ and $\gamma^{th}_2$ are corresponding predetermined thresholds.		
	
	The definition in (\ref{def_non_unifrom_XPD_aperture}) indicates that when the antenna aperture $A\ge A^{th}$, the ratio between the maximum and minimum values of the XPD component $\chi_1$ and that of $\chi_2$ are no smaller than predetermined thresholds.
	\end{definition}}
	
	Before providing the exact expression of the non-uniform XPD aperture, we first introduce the definitions of some parameters. Define parameter $b$ as
	\begin{align}
		\label{expression_b}
		b=\left|\sum_l \left.\dfrac{\partial \chi_2 }{\partial \phi_m^{(l)}}\right |_{\{\phi_m^{(l)}\}_l=\{\phi_0^{(l)}\}_l}\frac{-\delta^{(l)}\sqrt{1-(\delta^{(l)})^2}}{|\delta^{(l)}|r^{(l)}_0\sin\theta^{(l)}_0}\right|,
	\end{align}
	where $(r^{(l)}_0, \theta^{(l)}_0, \phi_0^{(l)})$ is the distance, elevation departure angle, and azimuth departure angle of the $l$-th cluster with respect to the $0$-th BS antenna, respectively. In addition, the expression of $\delta^{(l)}$ in (\ref{expression_b}) is given by $\delta^{(l)}=\frac{\cos\phi_0^{(l)}+k\sin\phi_0^{(l)}}{\sqrt{1+k^2}}$. 
	
	Then, based on Definition~\ref{def_different_XPD}, the exact expression of non-uniform XPD aperture $A^{th}$ can be acquired, as shown in the following theorem.
	{
	\begin{theorem}
		\label{theorem_minimum_array_size}
		Non-uniform XPD aperture $A^{th}$ can be expressed as
		\begin{align}
			\label{expression_aperture_th}
			A^{th}=\max\Bigg\{&\frac{k}{1+k^2}\left(\frac{2r_U}{\eta\delta_U}(1-\frac{2}{\gamma^{th}_1+1})\right),\notag\\
			&\frac{k}{1+k^2}\left(\frac{2\chi_2(\{\phi_0^{(l)}\}_l)(\gamma^{th}_2-1)}{(\gamma^{th}_2+1)b}\right)^2\Bigg\},
		\end{align}
		where $r_U$ is the distance between the UE and the center of the BS antenna, $\eta$ is an XPD parameter defined in (\ref{influence_distance}), $\chi_2(\{\phi_0^{(l)}\}_l)$ represents the value of term $\chi_2(\{\phi_m^{(l)}\}_l)$ corresponding to the $0$-th BS antenna, i.e., the one at the center of the array, and the definitions of $\delta_U$ and $b$ can be found in (\ref{delta_U}) and (\ref{expression_b}), respectively.
		
	\end{theorem}}
	\begin{proof}
		See Appendix~\ref{app_theorem_minimum_array_size}.
	\end{proof}
	\vspace{-0.2cm}
	
	\subsubsection{Impact of Non-Uniform XPD and Pathloss on Transmit Covariance Design}
	As mentioned in Remark~\ref{remark_shorter_XPD_distance_than_power_distance}, when the BS-UE distance is shorter than the derived non-uniform XPD distance $r_U^{th}$, the variations of XPDs and pathloss across BS antennas cannot be ignored, which, however, makes it challenging to design transmit covariance matrix, as discussed below. 
	\begin{lemma}
		\label{lemma_convetional_MIMO_identity}
		In conventional massive MIMO networks with dual-polarized antennas, the XPD and pathloss for different BS antennas can be considered equal. {Therefore, the average channel power gain is the same over different antenna elements, i.e., 
		\begin{align}
			\label{equal}
			\forall n, m,~\mathbb{E}(|\bm{G}^{(ji)}(n,m)|^2)=p^{(ji)},~j,i\in\{V,H\},
		\end{align}
		where $\bm{G}^{(ji)}$ is channel matrix given in (\ref{channel}), and {$p_{j,i}$ is the average channel power gain from one BS antenna in polarization $i$ to one UE antenna in polarization $j$}.
		
		The property in (\ref{equal}) is utilized by existing algorithm to simplify transmit covariance design}, i.e., the optimal transmit covariance that maximizes the ergodic capacity is found to be a scalar matrix, with diagonal elements given by $\frac{1}{2M}$. Here, $M$ represents the number of dual-polarized antennas.
	\end{lemma}
	\begin{proof}
		See Appendix~\ref{app_lemma_convetional_MIMO_identity}.
	\end{proof}
	
	\begin{remark}
		\label{theorem_conventional_not_apply}
		Unlike conventional dual-polarized massive MIMO, the XPD and pathloss of channels can vary across BS antennas in dual-polarized XL-MIMO systems. {This is because the increased aperture of the BS antenna array leads to enlarged non-uniform XPD distances. Therefore, users are more likely to fall into the XPD-based near-field region, where the XPD and pathloss variations across array elements cannot be overlooked.
		
		
		Due to the XPD variance, the property in (\ref{equal}) does not hold.} Consequently, the transmit covariance design scheme for the conventional dual-polarized massive MIMO, i.e., scalar matrix, cannot be directly applied\footnote{The channel matrix is not deterministic for the BS, since we assume that instantaneous CSI, which includes small-scale fading effects, is unknown to the BS. {Further, the statistical CSI is assumed to be available to the BS.}}, which makes it \textbf{challenging} to optimize transmit covariance matrices in dual-polarized XL-MIMO systems.
	\end{remark}
	
	To cope with this challenge, we formulate a transmit covariance optimization problem and propose an efficient algorithm accounting for the non-uniform XPD and non-uniform pathloss to solve this problem in the following two sections.
\vspace{0.1cm}
	\section{Transmit Covariance Optimization Problem Formulation and Decomposition}
	\vspace{0.05cm}
	\label{sec_problem_formulation_and_decomposition}
	In this section, we formulate a transmit covariance optimization problem for dual-polarized XL-MIMO networks, which is then decomposed into two sub-problems, i.e., eigenvector optimization subproblem and power allocation subproblem.
%
\vspace{-0.15cm}
	\subsection{Transmit Covariance Optimization Problem Formulation}
	We aim to improve the ergodic capacity in (\ref{ergodic_capacity}) by optimizing transmit covariance matrix $Q$. Formally, the optimization problem is formulated as,
	\begin{subequations}\label{opt_problem_covariance}
	\begin{align}
		\label{obj}
		&\max_{\bm{Q}} \mathbb{E}\{\log_2\det(\bm{I}_{2N}+\gamma \bm{G}\bm{Q}\bm{G}^{\dagger})\},\\
		\label{cons_semi_def}
		s.t.~&\bm{Q}\succeq \bm{0},\\
		\label{cons_max_sum_power}
		&\Tr(\bm{Q})=1,\\
		\label{cons_max_each_power}
		&\bm{Q}(i,i)\le q_0.		
	\end{align}
	\end{subequations}
	Since $\bm{Q}$ represents transmit covariance matrix, it is required to be semi-definite, as indicated in constraint (\ref{cons_semi_def}). Constraint (\ref{cons_max_sum_power}) shows that the sum of the normalized transmit power over all BS antennas, i.e., the sum of the diagonal elements of matrix $\bm{Q}$, should be equal to $1$. Further, constraint (\ref{cons_max_each_power}) is the transmit power limit for each single antenna element at the BS~\cite{Bjornson_small_cell}, with the power limit denoted by $q_0$. {The Shannon capacity upper bound, i.e., SVD precoding with water-filling power allocation, requires perfect knowledge of instantaneous MIMO channels at the transmitter, which, however, is assumed to be unavailable in this paper. Therefore, the Shannon capacity upper bound cannot be directly utilized to solve the covariance matrix optimization problem in (\ref{opt_problem_covariance}).}
	
	\vspace{-0.1cm}
	\subsection{Problem Decomposition}
	Note that the transmit covariance matrix $\bm{Q}$ is Hermitian. Therefore, it can be unitarily diagonalized~\cite{Gao_eigenmode_2009}, i.e., 
	\begin{align}
		\label{diagnalization}
		\bm{Q}=\bm{U}\bm{\Lambda}\bm{U}^{\dagger},
	\end{align}
	where $\bm{U}\in \mathbb{C}^{2M\times 2M}$ is a unitary matrix consisting of the eigenvectors of $\bm{Q}$, and $\bm{\Lambda} \in \mathbb{C}^{2M\times 2M}$ is a diagonal matrix. We would like to point out $\bm{\Lambda}$ can also be interpreted as power allocation over the BS antennas, which will be further explained in Section~\ref{subsubsection_unitary}. Motivated by (\ref{diagnalization}), the transmit covariance optimization problem in (\ref{opt_problem_covariance}) can be decomposed into the following two subproblems.
	\subsubsection{Eigenvector optimization}
	Given the power allocation matrix $\bm{\Lambda}$, the subproblem for optimizing the unitary eigenvectors-based matrix $\bm{U}$ can be written as
	\begin{subequations}\label{opt_problem_eigenvector}
		\begin{align}
			\label{obj_eigenvector}
			&\max_{\bm{U}} \mathbb{E}\{\log_2\det(\bm{I}_{2N}+\gamma \bm{G}\bm{U}\bm{\Lambda}\bm{U}^{\dagger}\bm{G}^{\dagger})\},\\
			\label{cons_semi_def_eigenvector}
			s.t.~&\bm{U}\bm{\Lambda}\bm{U}^{\dagger}\succeq \bm{0},\\
			\label{cons_max_sum_power_eigenvector}
			&\Tr(\bm{U}\bm{\Lambda}\bm{U}^{\dagger})=1,\\
			\label{cons_max_each_power_eigenvector}
			&\bm{U}\bm{\Lambda}\bm{U}^{\dagger}(i,i)\le q_0.		
		\end{align}
	\end{subequations}
	\subsubsection{Power allocation}
	Given the optimal unitary matrix $\bm{U}^*$, the power allocation subproblem can be written as
	\begin{subequations}\label{opt_problem_power_v0}
		\begin{align}
			\label{obj_power_v0}
			&\max_{\bm{\Lambda}} \mathbb{E}\{\log_2\det(\bm{I}_{2N}+\gamma \bm{G}\bm{U}^*\bm{\Lambda}(\bm{U}^*)^{\dagger}\bm{G}^{\dagger})\},\\
			\label{cons_semi_def_power_v0}
			s.t.~&\bm{U}^*\bm{\Lambda}(\bm{U}^*)^{\dagger}\succeq \bm{0},\\
			\label{cons_max_sum_power_power_v0}
			&\Tr(\bm{U}^*\bm{\Lambda}(\bm{U}^*)^{\dagger})=1,\\
			\label{cons_max_each_power_power_v0}
			&\bm{U}^*\bm{\Lambda}(\bm{U}^*)^{\dagger}(i,i)\le q_0.		
		\end{align}
	\end{subequations}
\section{Transmit Covariance Optimization Algorithm Design}
\label{sec_algorithm_design}
In this section, we first solve the eigenvector optimization problem in (\ref{opt_problem_eigenvector}). Then, by substituting the derived optimal unitary matrix into the power allocation problem in (\ref{opt_problem_power_v0}), we propose a power allocation algorithm to maximize the ergodic capacity, which accounts for the non-uniform XPD and pathloss. On this basis, an overall transmit covariance optimization algorithm is introduced.
\vspace{-.2cm}
	\subsection{Eigenvector Design}
	\label{subsubsection_unitary}
	The eigenvector design problem is well solvable for MIMO channels with a specific structure, as indicated in the following Lemma~\ref{lemma_unitary}.
	\begin{lemma}
		\label{lemma_unitary}
		When MIMO channels satisfy,
		\begin{align}
			\label{channel_model_rewrite}
			\bm{G}=\bm{U}_R(\bm{D}+\bm{M}\odot\bm{G}_{iid})\bm{U}_T^\dagger,
		\end{align}
		the optimal unitary matrix is $\bm{U}^*=\bm{U}_T$ without the per-element power constraint (\ref{cons_max_each_power_eigenvector}). 
		
		Here, $\bm{U}_T$ and $\bm{U}_R$ are $2M\times 2M$ and $2N\times 2N$ unitary matrices, $\bm{D}$ is a $2N\times 2M$ deterministic matrix with at most one nonzero element in each row and each column, $\bm{M}$ is a $2N\times 2M$ deterministic matrix with nonnegative elements, and $\bm{G}_{iid}$ is a $2N\times 2M$ random matrix with elements having zero mean and independent identical distributions~(i.i.d.)~\cite{Gao_eigenmode_2009}.
	\end{lemma}
	Motivated by Lemma~\ref{lemma_unitary}, we try to show that the considered dual-polarized XL-MIMO channels exhibit the structure described in (\ref{channel_model_rewrite}). Specifically, let
	\begin{align}
		\bm{U}_R=\bm{I}_{2N},\\
		\bm{U}_T=\bm{I}_{2M},\\
		\label{def_D}
		\bm{D}=\bm{0}_{2N\times 2M}.
	\end{align}
Further, by organizing the pathloss and small-scale fading, we can obtain matrix $\bm{M}$ and $\bm{G}_{iid}$, respectively, i.e.,
	\begin{align}
		\label{def_M}
		\bm{M}=\left[
		\begin{matrix}
			\bm{M}^{(VV)} & \bm{M}^{(VH)}  \\
			\bm{M}^{(HV)} & \bm{M}^{(HH)}
		\end{matrix}
		\right],\\
		\bm{G}_{iid}=\left[
		\begin{matrix}
			\bm{G}_{iid}^{(VV)} & \bm{G}_{iid}^{(VH)}  \\
			\bm{G}_{iid}^{(HV)} & \bm{G}_{iid}^{(HH)}
		\end{matrix}
		\right],
	\end{align}
	where, $\forall i,j\in\{V,H\}$, the $(n,m)$-th element of the matrix $\bm{M}^{(ji)}\in\mathbb{C}^{N\times M}$ is the square root of channel power gain from the $m$-th BS antenna in polarization $i$ to the $n$-th UE antenna in polarization $j$, i.e., $\sqrt{\mathbb{E}(|g^{(ji)}_{n,m}|^2)}$ while the $(n,m)$-th element of the matrix $\bm{G}_{iid}^{(ji)}$ is normalized channel $g^{(ji)}_{n,m}/\sqrt{\mathbb{E}(|g^{(ji)}_{n,m}|^2)}$. 
	
	Therefore, based on Lemma~\ref{lemma_unitary}, we can set the unitary matrix $\bm{U}$ equal to $\bm{U}_T$, i.e.,
	\begin{align}
		\label{optimal_unitary}
		(\bm{U})^*=\bm{U}_T=\bm{I}_{2M},
	\end{align}
	and the additional per-element power constraint in (\ref{cons_max_each_power_eigenvector}) can be tackled by selecting a proper diagonal matrix $\bm{\Lambda}$, as shown in Section~\ref{subsection_power_allocation}.
	
%
	
	Since the optimal unitary matrix is an identity matrix as indicated in (\ref{optimal_unitary}), we have $\bm{Q}=\bm{\Lambda}$. Note that the diagonal elements of covariance matrix $\bm{Q}$ correspond to the transmit power of the BS antennas. Therefore, $\bm{\Lambda}$ can be interpreted as a power allocation matrix, which will be optimized in the following section.

\vspace{-0.2cm}
	\subsection{Power Allocation Design}
	\label{subsection_power_allocation}
	By substituting (\ref{optimal_unitary}) into (\ref{opt_problem_power_v0}), the power allocation subproblem can be rewritten as
	\begin{subequations}\label{opt_problem_power}
		\begin{align}
			\label{obj_power}
			&\max_{\bm{\Lambda}} \mathbb{E}\{\log_2\det(\bm{I}_{2N}+\gamma \bm{G}\bm{\Lambda}\bm{G}^{\dagger})\},\\
			\label{cons_semi_def_power}
			s.t.~&\bm{\Lambda}\succeq \bm{0},\\
			\label{cons_max_sum_power_power}
			&\Tr(\bm{\Lambda})=1,\\
			\label{cons_max_each_power_power}
			&\bm{\Lambda}(m,m)\le q_0.		
		\end{align}
	\end{subequations}
	For simplicity, we rewrite the power allocation problem (\ref{opt_problem_power}) by uncovering the diagonal elements of power allocation $\bm{\Lambda}=\diag\{\lambda_1^{(V)},\dots,\lambda_{M}^{(V)},\lambda_1^{(H)},\dots,\lambda_{M}^{(H)}\}$, i.e., 
	\begin{subequations}\label{opt_problem_power_v2}
		\begin{align}
			\label{obj_power_v2}
			\max_{\{\lambda_{m}^{(V)},\lambda_m^{(H)}\}}&\mathbb{E}\{\log_2\det(\bm{I}_{2N}+\gamma \bm{G}\bm{\Lambda}\bm{G}^{\dagger})\},\\
			\label{cons_semi_def_power_v2}
			s.t.~&\lambda_{m}^{(V)},\lambda_m^{(H)} \ge 0,\\
			\label{cons_max_sum_power_power_v2}
			&\sum_m(\lambda_{m}^{(V)}+\lambda_m^{(H)})=1,\\
			\label{cons_max_each_power_power_v2}
			&\lambda_{m}^{(V)},\lambda_m^{(H)} \le q_0,	
		\end{align}
	\end{subequations}
	where $\lambda_m^{(V)}$ and $\lambda_m^{(H)}$ represent the transmit power of the $m$-th BS antenna in polarization $V$ and polarization $H$, respectively. Due to the expectation operation in the objective function, it is non-trivial to solve problem (\ref{opt_problem_power_v2}). 
	
	To cope with this issue, in the following, we first propose a tight and tractable upper bound using permanents\footnote{The existing permanent-based power allocation algorithm in~\cite{Gao_eigenmode_2009} cannot be directly applied to dual-polarized XL-MIMO systems. This is because the complexity of calculating the objective function, which includes permanent operations, is unacceptable due to the large number of BS antennas in dual-polarized XL-MIMO systems, as shown in Remark~\ref{remark_permanent_complexity}. To cope with this issue, we propose a method to simplify the calculation of the objective function as shown in Section~\ref{subsubsection_simplification_of_calculation}, which serves as the main contribution of the proposed algorithm compared with the existing algorithm in~\cite{Gao_eigenmode_2009}.}. Then a sub-array-based method is utilized to simplify the calculation of the capacity bound. Afterwards, a gradient-descent-based approach is proposed to solve the power allocation problem. 
	\subsubsection{\textit{Closed-form capacity upper bound using permanent}}
	\label{upper_bound_permanent}
	Before presenting the capacity upper bound, we first introduce the definition of the permanent operation, as the sequel.
	\begin{definition}
		For an $X\times Y$ matrix $\bm{A}$, the permanent is defined as~\cite{Gao_eigenmode_2009}
		\begin{align}
			\label{def_per}
			\permanent({\bm{A}})=\left\{
			\begin{aligned}
				&\sum_{\hat{\bm{\xi}}_X\in\bm{\Xi}_Y^X}\prod_{x=1}^X a_{x,\xi_x},~X\le Y,\\
				&\sum_{\hat{\bm{\xi}}_Y\in\bm{\Xi}_X^Y}\prod_{y=1}^Y a_{\xi_y,y},~X>Y,
			\end{aligned}
			\right.
		\end{align}
		where $\bm{\Xi}_Y^X$ denotes the set of all size-$X$ permutations of the numbers $\{1,2,\dots,Y\}$. The notation $\hat{\bm{\xi}}_X\in\bm{\Xi}_Y^X$ means that $\hat{\bm{\xi}}_X=(\xi_1,\xi_2,\dots,\xi_X)$, where $\xi_x\in\{1,2,\dots,Y\}$ and $\xi_{x_1}\neq\xi_{x_2}$. Further, $a_{x,\xi_x}$ denotes the $(x,\xi_x)$-th element of the matrix $\bm{A}$. The definitions of $\bm{\Xi}_X^Y$, $\hat{\bm{\xi}}_Y$, and $a_{\xi_y,y}$ are similar, and thus are omitted due to space limitations.
	\end{definition} 
	Based on the Jensen's inequality and permanents, a tight and closed-form upper bound can be derived for the ergodic capacity, as shown in the following theorem.
	\begin{theorem}
		The ergodic capacity can be upper bounded as~\cite{Gao_eigenmode_2009}
		\begin{align}
			\label{upper_bound}
			C\le \log_2\bigg(\permanent([\bm{I}_{2N},\gamma\bm{\Omega}\bm{\Lambda}])\bigg)\triangleq C^{ub}.
		\end{align}
		Here, $\bm{\Omega}=\bm{D}\odot \bm{D}+\bm{M}\odot\bm{M}$, where $\bm{D}$ and $\bm{M}$ are channel parameters defined in Section~\ref{subsubsection_unitary}.
	\end{theorem}
	Therefore, the power allocation problem can be further transformed into
	\begin{subequations}\label{opt_problem_power_v2dot5}
		\begin{align}
			\label{obj_power_v2dot5}
			\max_{\{\lambda_{m}^{(V)},\lambda_m^{(H)}\}}&C^{ub}=\log_2\bigg(\permanent([\bm{I}_{2N},\gamma\bm{\Omega}\bm{\Lambda}])\bigg),\\
			\label{cons_semi_def_power_v2dot5}
			s.t.~&\lambda_{m}^{(V)},\lambda_m^{(H)} \ge 0,\\
			\label{cons_max_sum_power_power_v2dot5}
			&\sum_m(\lambda_{m}^{(V)}+\lambda_m^{(H)})=1,\\
			\label{cons_max_each_power_power_v2dot5}
			&\lambda_{m}^{(V)},\lambda_m^{(H)} \le q_0.		
		\end{align}
	\end{subequations}
	
	\begin{remark}
		\label{remark_permanent_complexity}
		According to the definition of permanent in (\ref{def_per}), the number $\mathbb{N}_C^{ori}$ of  multiplications required for one calculation of the capacity upper bound $C^{ub}$ is
		\begin{align}
			\label{comlexity_direct}
			\mathbb{N}_C^{ori}=\frac{(2M+2N)!(2N-1)}{2M!},
		\end{align}
		where $M$ and $N$ represent the number of BS antennas and UE antennas, respectively.

		Therefore, the complexity for calculating the capacity upper bound $C^{ub}$ is \textit{unacceptable}, due to the large number of BS antenna elements in the considered dual-polarized XL-MIMO system. Consequently, it is \textit{challenging} to solve the optimization problem in (\ref{opt_problem_power_v2dot5}).
	\end{remark}
	As an example, we assume that the BS is equipped with $M=50$ dual-polarized antennas and the UE has $N=2$ dual-polarized antennas. Then, more than $10^8$ multiplications are required for one single calculation of the objective function.
	\subsubsection{\textit{Simplification of the calculation of capacity bound}}
	\label{subsubsection_simplification_of_calculation}
	To simplify the calculation of the capacity bound, we propose a sub-array-based method. The main idea is to divide the entire antenna array at the BS into several sub-arrays. Since the aperture of each sub-array is much smaller than that of the overall antenna array, the XPD and pathloss corresponding to different BS antennas within each sub-array can be regarded as the same, and thus the transmit power should be equally allocated within each sub-array, which can be utilized to simplify calculations of the objective function. In the following, we will elaborate on the sub-array-based method and the corresponding computational complexity. 
	
	Specifically, the permanent in the objective function can be rewritten as~\cite{Gao_eigenmode_2009}
	\begin{align}
		\label{extended_permanent}
		\permanent([\bm{I}_{2N},\gamma\bm{\Omega}\bm{\Lambda}])=\sum_{k=0}^{2N} \sum_{\bm{\xi}_k\in\Xi^{(k)}_{2M}}\permanent\left(\left(\gamma\bm{\Omega}\bm{\Lambda}\right)_{\bm{\xi}_k}\right),
	\end{align}
	where $\Xi^{(k)}_{2M}$ is the set of all ordered length-$k$ subsets of the numbers $\{1,2,\dots,2M\}$. By the notation $\bm{\xi}_k\in\Xi^{(k)}_{2M}$, we mean that $\bm{\xi}_k=(x_1,x_2,\dots,x_{k})$, where $x_m\in \{1,2,\dots,2M\}$ for $1\le m\le k$, and $x_1<x_2<\cdots<x_{k}$. Besides, $(\bm{X})_{\bm{\xi}_k}$ indicates the sub-matrix of $\bm{X}$ obtained by selecting the columns indexed by $\bm{\xi}_k$. In addition, we define $\permanent\left(\left(\gamma\bm{\Omega}\bm{\Lambda}\right)_{\bm{\xi}_k}\right)=1$ when $k=0$.
	
	Motivated by (\ref{extended_permanent}), the calculation of the permanent can be simplified by reformulating the matrix $\gamma\bm{\Omega}\bm{\Lambda}$ and the set $\Xi^{(k)}_{2M}$. First, the matrix $\gamma\bm{\Omega}\bm{\Lambda}$ is reformulated. Denote the number of sub-arrays and the number of dual-polarized antenna elements within each sub-array as $S$ and $M_0$, respectively, and thus we have $S\times M_0=M$. By substituting the definitions of $\bm{M}$ and $\bm{D}$ in (\ref{def_M}) and (\ref{def_D}) into $\gamma\bm{\Omega}\bm{\Lambda}$, we have
	\begin{align}\label{matrix_decompose}
		\gamma\bm{\Omega}\bm{\Lambda}
		=\left[
		\begin{matrix}
			\bm{A}_{1}^{(V)} &\dots& \bm{A}_{S}^{(V)}&  \bm{A}_{1}^{(H)} &\dots& \bm{A}_{S}^{(H)}
		\end{matrix}
		\right],
	\end{align}
	where, $\forall j\in\{V,H\}$, the definition of $\bm{A}_{s}^{(j)}$ can be given by
	\begin{align}\label{submatrix}
		\bm{A}_{s}^{(j)}
		=\left[
		\begin{matrix}
			\gamma\beta_{(s-1)M_0+1}^{(Vj)}\lambda_{(s-1)M_0+1} &\dots&\gamma\beta_{sM_0}^{(Vj)}\lambda_{sM_0}\\ \vdots&&\vdots\\
			\gamma\beta_{(s-1)M_0+1}^{(Vj)}\lambda_{(s-1)M_0+1} &\dots&\gamma\beta_{sM_0}^{(Vj)}\lambda_{sM_0}\\
			\gamma\beta_{(s-1)M_0+1}^{(Hj)}\lambda_{(s-1)M_0+1} &\dots&\gamma\beta_{sM_0}^{(Hj)}\lambda_{sM_0}\\
			\vdots&&\vdots\\
			\gamma\beta_{(s-1)M_0+1}^{(Hj)}\lambda_{(s-1)M_0+1} &\dots&\gamma\beta_{sM_0}^{(Hj)}\lambda_{sM_0}
		\end{matrix}
		\right].
	\end{align}
	Here, antenna elements $\{(s-1)M_0+1,\dots,sM_0\}$ belong to the $s$-th sub-array. Due to the small aperture of the $s$-th sub-array, the XPD and pathloss are approximately the same among the antenna elements within the sub-array, i.e., for any $m_1,m_2\in \{(s-1)M_0+1,\dots,sM_0\}$,
	\begin{align}
		\label{XPD_equal}
		XPD_{m_1}&\approx XPD_{m_2},\\
		\label{pathloss_equal}
		\beta_0d_{m_1}^{-\alpha}&\approx \beta_0d_{m_2}^{-\alpha}.
	\end{align}
	Therefore, the transmit power should be equally allocated among the BS antennas within the same sub-array, i.e.,
	\begin{align}
		\label{power_equal}
		\lambda_{m_1}^{(j)}=\lambda_{m_2}^{(j)},~j\in\{V, H\}.
	\end{align}
	Based on (\ref{XPD_equal})-(\ref{power_equal}), we can find that the columns of matrix $\bm{A}_{s}^{(j)}$ are the same. The above discussions can be summarized as the following remark.
	\begin{lemma}
		\label{lemma_same_column}
	The matrix $\gamma\bm{\Omega}\bm{\Lambda}$ can be rewritten as a collection of $2S$ submatrices as shown in (\ref{submatrix}), where for each submatrix $\bm{A}_{s}^{(j)}$, all its columns are the same.
	\end{lemma}
	

	Then, we move on to the reformulation of set $\Xi^{(k)}_{2M}$. Based on (\ref{matrix_decompose}) and Lemma~\ref{lemma_same_column}, we can find that for each given $\bm{\xi}_k$, 
	\begin{align}
		\label{same_value}
\exists \bm{\xi}_{k}^{'}, s.t.,\left(\gamma\bm{\Omega}\bm{\Lambda}\right)_{\bm{\xi}_k}=\left(\gamma\bm{\Omega}\bm{\Lambda}\right)_{\bm{\xi}_{k}^{'}}
	\end{align}
	Therefore, we can approximate $\permanent\left(\left(\gamma\bm{\Omega}\bm{\Lambda}\right)_{\bm{\xi}_{k}^{'}}\right)$ with $\permanent\left(\left(\gamma\bm{\Omega}\bm{\Lambda}\right)_{\bm{\xi}_k}\right)$, without the need of additional calculations. Motivated for this property, in the following, we will try to classify the set $\Xi^{(k)}_{2M}$ into several subsets, where the elements $\bm{\xi}_k$ within the same subset satisfy (\ref{same_value}). 
	
	Define $\bm{b}(\bm{\xi}_k)=\{b_1^{(V)},\dots,b_S^{(V)},b_1^{(H)},\dots,b_S^{(H)}\}$, where $b_s^{(j)}$ indicates how many column vectors of submatrix $\bm{A}_{s}^{(j)}$ defined in (\ref{submatrix}) are included in matrix $\left(\gamma\bm{\Omega}\bm{\Lambda}\right)_{\bm{\xi}_k}$. Therefore, given $k$, there are in total $\frac{(2S+k)!}{(k+1)!(2S)!}$ possible values for $\bm{b}$, where $S$ is the number of sub-arrays. The set $\Xi^{(k)}_{2M}$ can thus be categorized into the same number of subsets accordingly. The $l$-th subset corresponds to the $l$-th value $\bm{b}^{(l)}=\{b_1^{(V),(l)},\dots,b_S^{(V),(l)},b_1^{(H),(l)},\dots,b_S^{(H),(l)}\}$ of $\bm{b}$, and can be given by
	\begin{align}
		\label{def_set}
		\Xi^{(k),(l)}_{2M}=\{\bm{\xi}_k|\bm{\xi}_k\in \Xi^{(k)}_{2M}, \bm{b}(\bm{\xi}_k)=\bm{b}^{(l)}\}.
	\end{align}
	with the cardinality of this subset given by
	\begin{align}
		|\Xi^{(k),(l)}_{2M}|=\prod_s\frac{(M_0!)^2}{(M_0-b_s^{(V),(l)})!b_s^{(V),(l)}!(M_0-b_s^{(H),(l)})!b_s^{(H),(l)}!},
	\end{align}
	Based on (\ref{same_value})-(\ref{def_set}), we have the following lemma.
	\begin{lemma}
		\label{lemma_different_set}
		The set $\Xi^{(k)}_{2M}$ can be categorized into $L$ subsets defined in (\ref{def_set}), where each $\bm{\xi}_k$ and $\bm{\xi}_{k}^{'}$ belonging to $\Xi^{(k),(l)}_{2M}$ satisfy (\ref{same_value}).
	\end{lemma}
	Based on (\ref{extended_permanent}), Lemma~\ref{lemma_same_column}, and Lemma~\ref{lemma_different_set}, the calculation of capacity upper bound $C^{ub}$ can be simplified, as shown in the following theorem.
	\begin{theorem}
		\label{the_simplify}
		The capacity upper bound $C^{ub}$ in (\ref{obj_power_v2dot5}) can be written as
		\begin{align}
			C^{ub}=\log_2\left(f(\bm{\Lambda})\right),
		\end{align}
		where $f(\bm{\Lambda})$ is given by
		\begin{align}
			\label{simplify}
			f(\bm{\Lambda})&\triangleq\sum_{k=0}^{2N} \sum_{\bm{\xi}_k\in\Xi^{(k)}_{2M}}\permanent\left(\left(\gamma\bm{\Omega}\bm{\Lambda}\right)_{\bm{\xi}_k}\right)\notag\\
			&=\sum_{k=0}^{2N}\sum_{l=1}^{\frac{(2S-1)!}{(2S-k)!(2k-1)!}} |\Xi^{(k),(l)}_{2M}|\permanent\left(\left(\gamma\bm{\Omega}\bm{\Lambda}\right)_{\bm{\xi}_k^{(l)}}\right),
		\end{align}
		where $\bm{\xi}_k^{(l)}$ is a vector arbitrarily selected from the set $\Xi^{(k),(l)}_{2M}$.
	\end{theorem}
	\begin{remark}
		Based on (\ref{simplify}), the calculation of the capacity upper bound $C^{ub}$ can be simplified, with the number $\mathbb{N}_C^{sim}$ of multiplications required for one calculation of $C^{ub}$ given by
		\begin{align}
			\label{computation_complexity_simplify}
			\mathbb{N}_C^{sim}=\sum_{k=0}^{2N}(\frac{(k-1)2N!}{(2N-k)!}+1)(\frac{(2S+k)!}{(k+1)!(2S)!}).
		\end{align}
	\end{remark}
	
	The following remark compares the complexity $\mathbb{N}_C^{sim}$ of the simplified method from Theorem~\ref{the_simplify} against the complexity $\mathbb{N}_C^{ori}$ of the definition-based method noted in Remark~\ref{remark_permanent_complexity}.
	\begin{remark}
		\label{remark_complexity_bound}
	(\textbf{Complexity comparison}) To demonstrate the effectiveness of the proposed calculation method in Theorem~\ref{the_simplify}, we provide an upper bound on its complexity $\mathbb{N}_C^{sim}$, i.e.,
	\begin{align}
		\label{upper_bound_computation_complexity}
		\mathbb{N}_C^{sim}\le \frac{(2S+2N)!(2N-1)}{2S!}.
	\end{align}
	
	By comparing (\ref{upper_bound_computation_complexity}) and the complexity $\mathbb{N}_C^{ori}$ in (\ref{comlexity_direct}), we can see that since the number $M$ of BS antennas is much larger than the number $S$ of subarrays, i.e., $M\gg S$, the proposed method can effectively simplify the calculation of the capacity upper bound.
	\end{remark}
	\begin{proof}
		See Appendix~\ref{app_remark_complexity_bound}.
	\end{proof}
	
	\subsubsection{\textit{Power allocation algorithm}}
	Recall that the transmit power is equally allocated among different antenna elements within the same subarray, as indicated in (\ref{power_equal}). Therefore, we can use one single variable $\widetilde{\lambda}_s^{(j)}$ to represent the transmit powers of all antenna elements in polarization $j\in\{V,H\}$ within the $s$-th subarray, i.e.,
	\begin{align}
		\label{power_simplify}
		\lambda_m^{(j)}=\widetilde{\lambda}_s^{(j)}, \forall m\in\{(s-1)M_0+1,\dots,sM_0\},
	\end{align}
	where $M_0$ represents the number antenna elements within one subarray. Based on (\ref{extended_permanent}), (\ref{simplify}), and (\ref{power_simplify}), the power allocation problem in (\ref{opt_problem_power_v2dot5}) can be simplified as
\begin{subequations}\label{opt_problem_power_v3}
	\begin{align}
		\label{obj_power_v3}
		\max_{\{\widetilde{\lambda}_s^{(V)},\widetilde{\lambda}_s^{(H)}\}}&\log_2\left(f(\{\widetilde{\lambda}_s^{(V)},\widetilde{\lambda}_s^{(H)}\})\right),\\
		\label{cons_semi_def_power_v3}
		s.t.~&\widetilde{\lambda}_s^{(V)},\widetilde{\lambda}_s^{(H)} \ge 0,\\
		\label{cons_max_sum_power_power_v3}
		&\sum_s(\widetilde{\lambda}_s^{(V)}+\widetilde{\lambda}_s^{(H)})=\frac{1}{S},\\
		\label{cons_max_each_power_power_v3}
		&\widetilde{\lambda}_s^{(V)},\widetilde{\lambda}_s^{(H)} \le q_0.		
	\end{align}
\end{subequations}
	
	To deal with the constrained optimization problem in (\ref{opt_problem_power_v3}), \textit{the penalty method} is applied~\cite{Boyd_2004,Zeng_coverage_2021}. The penalty method involves transforming the original constrained optimization problem into an unconstrained one by adding penalty terms to the objective function, which penalize violations of the constraints. The penalties are typically weighted by a parameter $\mu$ that controls their influence on the optimization process. As the optimization progresses, the penalties encourage the optimizer to find solutions that satisfy the constraints. 
	
	More specifically, the $o$-th unconstrained optimization problem given the penalty parameter $\mu_o$ can be expressed as
	\begin{align}\label{opt_problem_power_v4}
		\max_{\{\widetilde{\lambda}_s^{(V)},\widetilde{\lambda}_s^{(H)}\}}&\log_2\left(f(\{\widetilde{\lambda}_s^{(V)},\widetilde{\lambda}_s^{(H)}\})\right)+p_E(\mu_o,\{\widetilde{\lambda}_s^{(V)},\widetilde{\lambda}_s^{(H)}\}).		
	\end{align}
	Here, $p_E(\mu_o,\{\widetilde{\lambda}_s^{(V)},\widetilde{\lambda}_s^{(H)}\})$ is the second-order penalty function, whose explicit expression is given in (\ref{penalty_function}) shown at the bottom of this page.
	\begin{figure*}[!hb]
		\rule[-12pt]{17.5cm}{0.05em}
		\begin{equation}
			\setlength{\abovedisplayskip}{3pt}
			\setlength{\belowdisplayskip}{-3pt}
			\begin{aligned}
				\label{penalty_function}
				p_E(\mu_o,\{\widetilde{\lambda}_s^{(V)},\widetilde{\lambda}_s^{(H)}\})=\mu_o\left(\left(\sum_s(\widetilde{\lambda}_s^{(V)}+\widetilde{\lambda}_s^{(H)})-\frac{1}{S}\right)^2+\sum_{j\in\{V,H\}}\sum_s\left(\max\{\widetilde{\lambda}_s^{(j)}-q_0,0\}\right)^2+\left(\max\{\widetilde{-\lambda}_s^{(j)},0\}\right)^2\right),
			\end{aligned}
		\end{equation}
	\end{figure*}
	Then, we utilize the gradient descent method to solve problem (\ref{opt_problem_power_v4}).
\begin{algorithm}[!tpb]{
		\caption{Transmit Covariance Optimization Algorithm}
		\begin{algorithmic}[1]\label{algorithm_covariance}
			\REQUIRE{Statistical channel information $\bm{\Omega}$, transmit SNR $\gamma$, user distance $r_U$, and non-uniform XPD distance $r_U^{th}$}
			\IF{$r_U>r_U^{th}$}
			\STATE{Set a scalar matrix as the optimal transmit covariance matrix, i.e., $\bm{Q}=\frac{1}{2M}\bm{I}_{2M}$.}
			\ELSE
			\STATE{Initialize $\{(\widetilde{\lambda}_s^{(V)})^*_{0},(\widetilde{\lambda}_s^{(H)})^*_{0}\}=\frac{1}{2S}$}
			\REPEAT 
			\STATE Set $\{(\widetilde{\lambda}_s^{(V)})^*_{o-1},(\widetilde{\lambda}_s^{(H)})^*_{o-1}\}$ as the initial solution;
			\STATE Given a penalty parameter $\mu_o$, solve problem (\ref{opt_problem_power_v4}) using gradient descent. Denote the optimal solution by $\{(\widetilde{\lambda}_s^{(V)})^*_{o},(\widetilde{\lambda}_s^{(H)})^*_{o}\}$;
			\STATE Let $\mu_{o+1}>\mu_o$, and $o=o+1$.
			\UNTIL {The objective function in (\ref{obj_power_v3}) converges and the constraints (\ref{cons_semi_def_power_v3})-(\ref{cons_max_each_power_power_v3}) are satisfied.}
			\STATE Acquire the optimal transmit covariance by substituting (\ref{optimal_unitary}) and the derived optimal power allocation into (\ref{diagnalization}).
			\ENDIF
			\ENSURE{Optimal transmit covariance matrix $\bm{Q}$}
	\end{algorithmic}}
\end{algorithm}

\subsection{Overall Algorithm}
Based on the method presented in the previous two subsections, we design a transmit covariance optimization algorithm to solve problem (\ref{opt_problem_covariance}). Specifically, in each iteration, gradient descent is used to solve problem (\ref{opt_problem_power_v4}). In the subsequent iteration, the optimal power allocations generated in the previous iteration are set as the initial solution. Further, a larger penalty parameter $\mu$ is substituted into problem (\ref{opt_problem_power_v4}) to encourage the optimizer to find solutions satisfying the constraints of (\ref{opt_problem_power_v3}). The iterations are completed until the difference in the objective function in (\ref{obj_power_v3}) between two adjacent iterations is less than a predetermined threshold and the constraints of (\ref{opt_problem_power_v3}) are satisfied. Finally, by substituting the derived optimal power allocations and the optimal eigenvectors in (\ref{optimal_unitary}) into (\ref{diagnalization}) can yield the optimal transmit covariance that maximizes the ergodic capacity. The algorithm is summarized in Algorithm~\ref{algorithm_covariance}. {From Algorithm~\ref{algorithm_covariance}, we can see that even given the detailed XPD and pathloss, the derived non-uniform XPD distance is still required for the optimization of the transmit covariance, i.e., they serve as a standard for deciding whether to use the proposed transmit covariance optimization algorithm or the simpler conventional method.}

{Further, we would like to point out that the following two characteristics for XPD-based near-field procedures are utilized in the optimization. First, when the user is within the derived non-uniform XPD distance, the variance of the pathloss and XPD over BS antenna elements cannot be overlooked, which serves as the motivation for simplifying the calculation of capacity bound, as pointed out in Section~\ref{upper_bound_permanent}. Second, when the user is in the XPD-based near-field, if we divide the whole antenna array into several smaller subarrays, such that each subarray is smaller than the non-uniform XPD aperture, the XPD and pathloss corresponding to different antenna elements within the same subarray can be regarded as the same. This characteristic serves as the main idea for simplifying the calculation of the capacity upper bound, as shown in Theorem~\ref{the_simplify} of the paper.}

{
\subsection{Complexity and Convergence Analysis}
Define \( O \) as the total number of iterations in Algorithm~1, i.e., the number of unconstrained optimization problems to be solved. Let \( T_{\max} \) denote the maximum number of iterations required to solve each unconstrained optimization problem using the gradient descent method, and we assume that each iteration contains \( V \) counts of the objective function in (\ref{computation_complexity_simplify}). The computational complexity of Algorithm~1 is then upper bounded by $\mathcal{O}(OT_{max}VM^N)$, where \( \mathbb{N}_C^{\text{sim}} \) denotes the number of multiplications required for one evaluation of the objective function in (\ref{computation_complexity_simplify}), and $M$ and $N$ represent the number of dual-polarized antennas at the BS and the UE, respectively.  

Furthermore, the convergence of the proposed algorithm can be guaranteed, as its framework is fundamentally based on the penalty function method~\cite{J_1999}.}


%
%

\begin{figure}[!tpb]
	\centering
	\vspace{0pt}
	\includegraphics[width=0.33\textwidth]{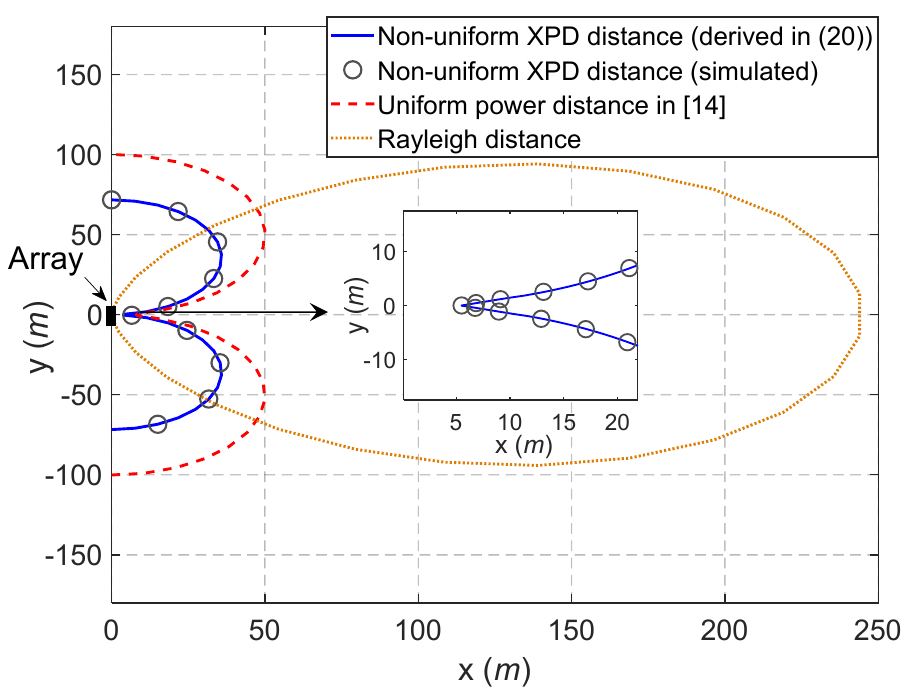}
	\vspace{-0.1cm}
	\caption{Comparison of non-uniform XPD distances $r_U^{th}$ with direction-dependent Rayleigh distances~\cite{Lu_XL_MIMO_ICC}. Further, non-uniform power distances are also considered, characterizing the maximum BS-user distance within which the difference in pathloss (and thus received power) across antenna elements cannot be neglected~\cite{Lu_XL_MIMO_ICC}. Here, the decay exponent of XPD with distance is set to $\eta=1$, and a linear antenna array with $70$ antennas is considered.}
	\vspace{-0.6cm}
	\label{fig_cmp_distance_threshold}
\end{figure}

\begin{figure}[!tpb]
	\centering
	\vspace{0pt}
	\includegraphics[width=0.35\textwidth]{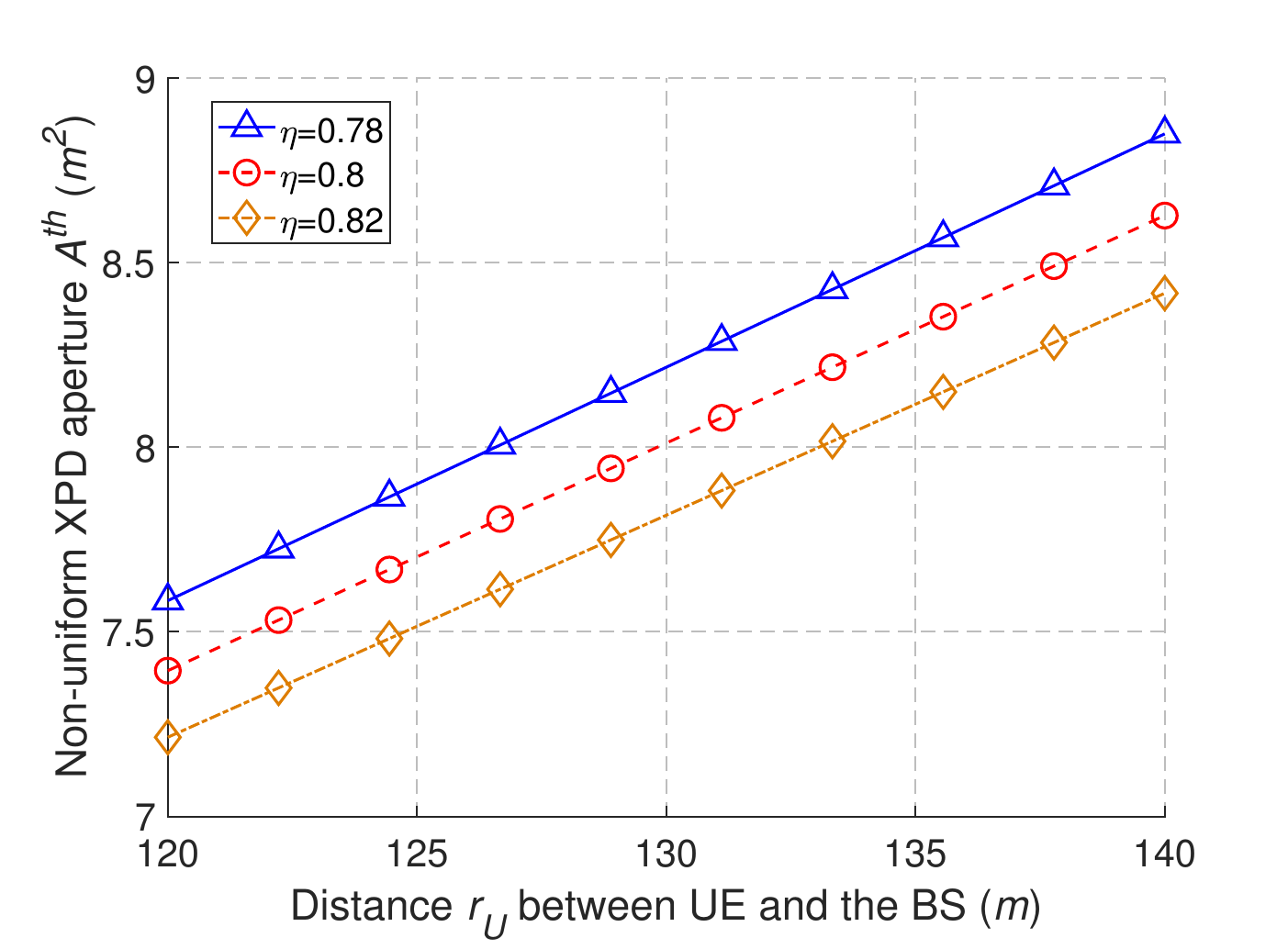}
	\vspace{-0.1cm}
	\caption{Non-uniform XPD aperture $A^{th}$ defined in (\ref{def_non_unifrom_XPD_aperture}) vs. the distance $r_U$ between the UE and the center of the BS antenna array, with elevation and azimuth angles of the user $(\theta_U,\varphi_U)=(\frac{\pi}{6},\frac{\pi}{2})$. The BS antenna array is square, i.e., $k=1$. The thresholds for XPD components $\chi_1$ and $\chi_2$ are given by $\gamma_1^{th}=\gamma_2^{th}=1.1$.}
	\vspace{-0.6cm}
	\label{fig_aperture_vs_r_U}
\end{figure}

\section{Simulation Results}
\label{sec_simulation}
To validate the theoretical analysis and the effectiveness of the proposed transmit covariance optimization algorithm, we evaluate the performance of the dual-polarized XL-MIMO system via simulations. The simulation parameters are set up based on the existing works~\cite{yu_dual_polarized_2022,Yue_IOS_2023,Zhang_IOS_mag}. Specifically, the UE is equipped with $N=2$ dual-polarized antennas. The wavelength is set to $\lambda=10$~cm and we assume half wavelength for antenna spacing. $L=5$ clusters are considered, whose coordinates are given by $(29.3, 0, 6.2)$~m, $(24.6, -4.3, 1.3)$~m, $(39, -9, 0)$~m, $(32.7, -2.9, -2.9)$~m, and $(48.5, -8.7, -8.6)$~m, respectively. The pathloss model is given by $\beta_m=\beta_0d_m^{-\alpha}$~\cite{Zeng_multi_user_RRS}, where the pathloss exponent is assumed to be $\alpha=4$, $d_m$ represents the distance from the $m$-th BS antenna to the UE, $\beta_0$ is pathloss at unit distance.  The XPD under unit Tx-Rx distance is set to $XPD_{d_m=1}=5$~dB, and the decay exponent for XPD with respect to propagation distances is given by $\eta=0.8$. The azimuth spread and truncation spread are assumed to be equal among the clusters, i.e., $\omega^{(l)}=35^\circ$ and $\Delta\varphi^{(l)}=180^\circ$ for all $l$. The thresholds for XPD components $\chi_1$ and $\chi_2$ are given by $\gamma_1^{th}=\gamma_2^{th}=1.05$, respectively. We assume Nakagami-fading for the XL-MIMO channels, with fading parameters $\mu=5$ and $\Omega$ equal to corresponding pathloss $\beta_m^{(ji)}$. Further, the amplitudes and phases of the channels are assumed to be independent, where the channel phases obey uniform distributions within $[0,2\pi)$. The overall transmit power and the variance of received AWGN is set to $43$~dBm and $\sigma^2=-96$~dBm, respectively. The ratio of the maximum transmit power per antenna to the total transmit power is set to $q_0=4\times(2M)^{-1}$.

{In Fig.~\ref{fig_cmp_distance_threshold}, we validate the theoretically derived non-uniform XPD distance $r_U^{th}$ against numerical results.} Further, the non-uniform XPD distance is also compared against existing direction-dependent Rayleigh distances and uniform power distances~\cite{Lu_XL_MIMO_2022}. {Here, the numerical results are acquired by adopting the interior-point method (by calling Matlab function \emph{fmincon}) to solve the optimization problem that defines the non-uniform XPD distance, as indicated in~(\ref{def_non_unifrom_XPD_distance}).} Further, the Rayleigh distances serve as conventional ways of separating near-field and far-field regions while uniform power distances characterize the maximum BS-user distance within which the difference in pathloss (and thus received power) cannot be neglected. {The threshold for the ratio between the maximum and minimum pathloss is set as $1.15$~\cite{Lu_XL_MIMO_ICC,Lu_XL_MIMO_2022}}. {From Fig.~\ref{fig_cmp_distance_threshold}, we can observe that the derived non-uniform XPD distances are consistent with numerical results.} Further, we can find that Rayleigh distances have different shapes compared to non-uniform XPD distances, since Rayleigh distances only take phase difference across array elements into account. Therefore, Rayleigh distances are insufficient to describe the discrepancies in XPD and pathloss in dual-polarized XL-MIMO systems, highlighting the importance of the derived non-uniform XPD distances. {More specifically, Fig.~\ref{fig_cmp_distance_threshold} shows that unlike classic Rayleigh distances, if non-uniform XPD distances are used for a refined near-far field boundary, users with inclined directions are more likely to be located within the near-field region compared to those in the normal direction. Further, it can be observed that non-uniform XPD distances are shorter than uniform power distances, which is consistent with Remark~\ref{remark_shorter_XPD_distance_than_power_distance}.}

{Fig.~\ref{fig_aperture_vs_r_U} shows the influence of the distance $r_U$ from the UE to the center of the BS antenna array on the non-uniform XPD aperture $A^{th}$ given in (\ref{def_non_unifrom_XPD_aperture}). From Fig.~\ref{fig_aperture_vs_r_U}, we can see that $A^{th}$ increases with $r_U$. This is because as the UE moves away from the BS, the discrepancies in BS-UE distances across different BS antennas is less significant. Therefore, to increase such discrepancies, a larger aperture is required so as to achieve the predetermined threshold $\gamma_1^{th}$ given in (\ref{def_non_unifrom_XPD_aperture}). Further, it can be found that the $A^{th}$ decreases with the decay exponent $\eta$ of XPD. This is because as $\eta$ becomes larger, the difference of the XPD between two antenna elements becomes more significant, and thus a smaller aperture is enough to achieve the predetermined XPD threshold.}



Fig.~\ref{fig_capacity_3_cmp} compares the ergodic capacity achieved by the proposed transmit covariance optimization algorithm with that achieved by the existing transmit covariance design scheme for the conventional dual-polarized massive MIMO, i.e., scalar matrix. From Fig.~\ref{fig_capacity_3_cmp}, we can observe that the proposed algorithm always achieve higher capacity than the scalar matrix. This indicates that the conventional scalar-matrix-based transmit covariance design scheme is not optimal for the dual-polarized XL-MIMO system, which verifies Theorem~\ref{theorem_conventional_not_apply}. Fig.~\ref{fig_capacity_3_cmp} also shows that the ergodic capacity achieved by the proposed algorithm increase as the size of sub-array becomes smaller. This is because there are more sub-arrays, leading to higher design degree of freedom. 

\begin{figure}[!tpb]
	\centering
	\vspace{0pt}
	\includegraphics[width=0.34\textwidth]{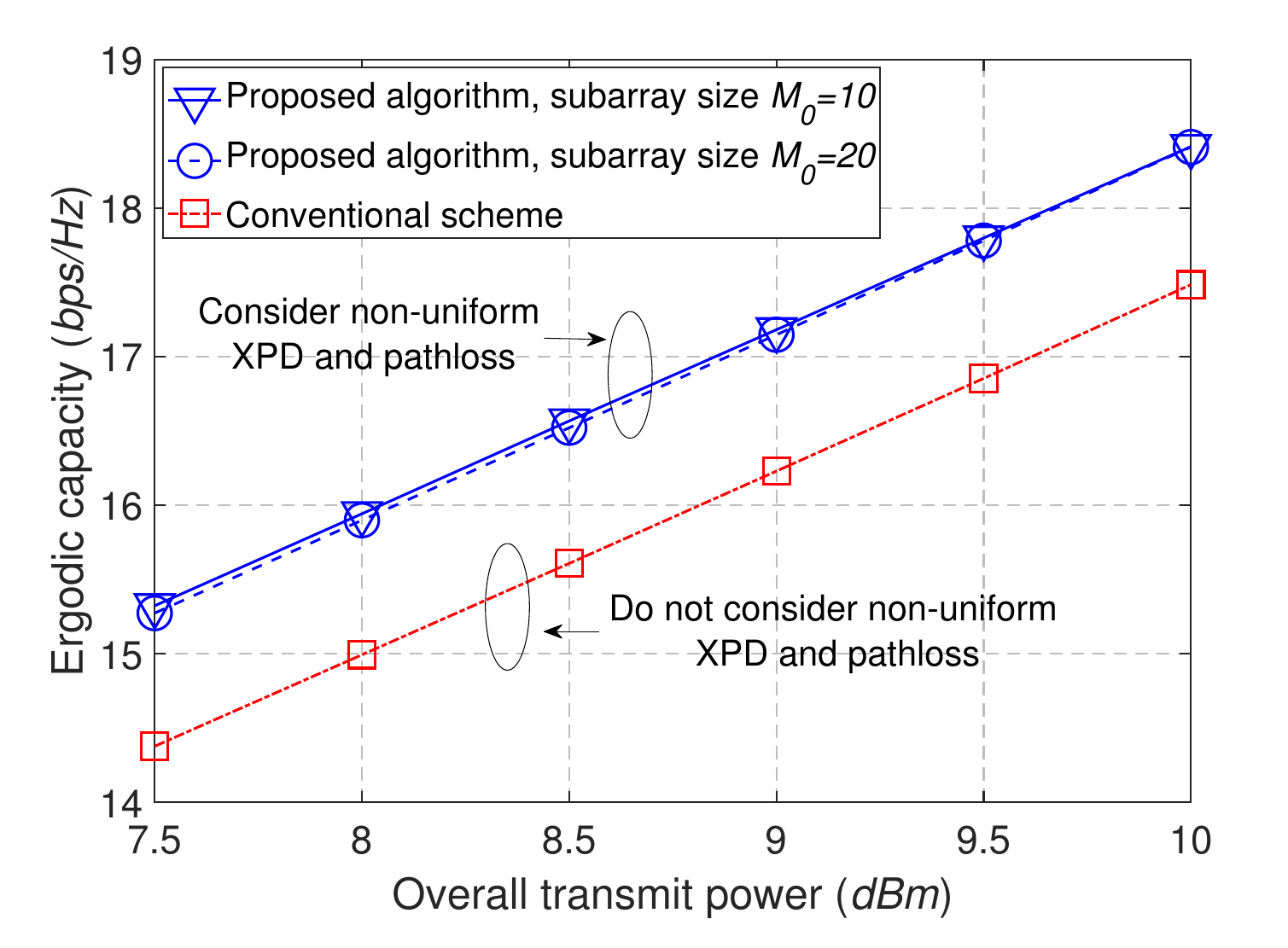}
	\vspace{-0.1cm}
	\caption{Ergodic capacity vs. overall transmit power $P$ of the BS, with the coordinate of the UE given by $(30,0,0)$~m, number of BS antennas $M=80$. A linear antenna array is deployed at the BS. Here, the ``Conventional scheme" refers to the one designed for conventional dual-polarized massive MIMO, which does not consider XPD and pathloss differences across array elements.}
	\vspace{-0.2cm}
	\label{fig_capacity_3_cmp}
\end{figure}

\begin{figure}[!tpb]
	\centering
	\vspace{0pt}
	\includegraphics[width=0.33\textwidth]{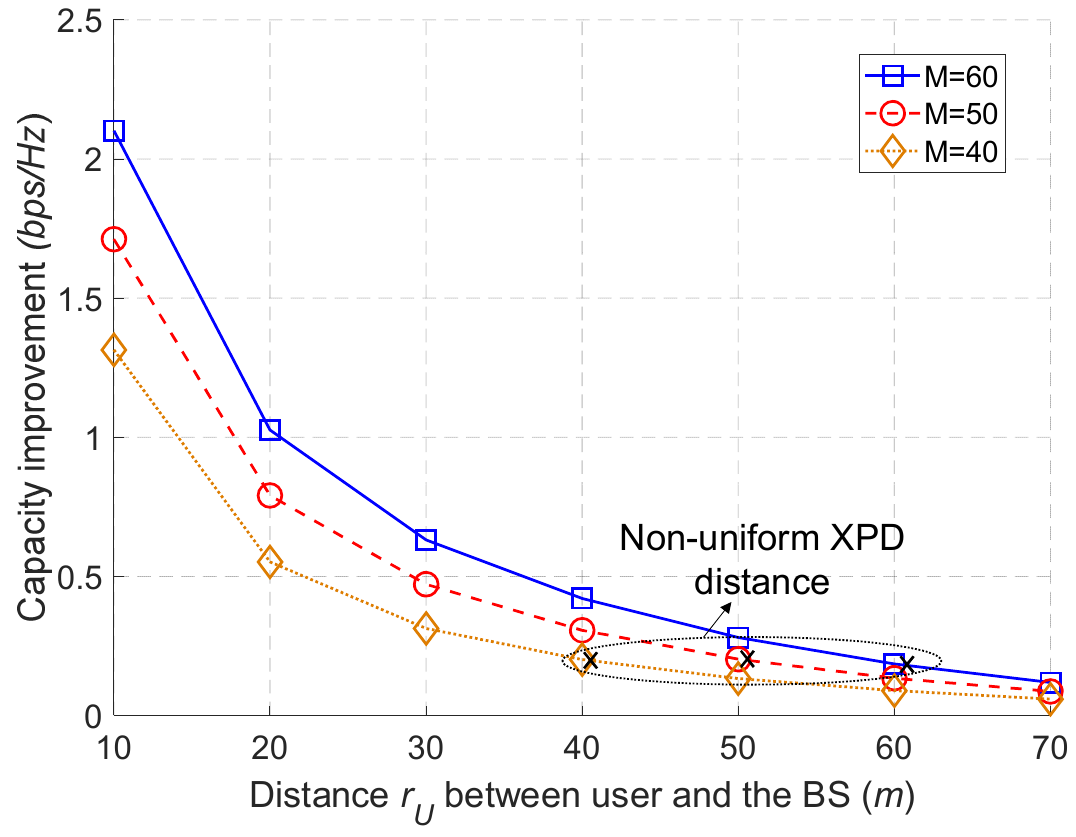}
	\vspace{-0.1cm}
	\caption{Improvement in ergodic capacity achieved by utilizing the non-uniform XPD model compared to the conventional model. Here, the decay exponent of XPD with distance is set to $\eta=1$. The thresholds for XPD components $\chi_1$ and $\chi_2$ are given by $\gamma_1^{th}=\gamma_2^{th}=1.05$. A linear antenna array is considered, and the user is set to be aligned with the antenna array.}
	\vspace{-0.4cm}
	\label{fig_capacity_improvement_vs_distance}
\end{figure}

{Specifically, Fig.~\ref{fig_capacity_improvement_vs_distance} illustrates the impact of the propagation distance $r_U$ on the improvement in ergodic capacity achieved by adopting the non-uniform XPD model, in comparison to the conventional model. From Fig.~\ref{fig_capacity_improvement_vs_distance}, we can find that when the user is relatively close the BS and is within the non-uniform XPD distance, the proposed algorithm achieves notable capacity enhancement compared with conventional methods, which demonstrates the effectiveness of the proposed algorithm. In contrast, for users far away from the BS, the capacity improvement is, however, less significant. This is because the XPD variance over the BS antennas is less significant. Especially, when the propagation distance is larger than the derived non-uniform XPD threshold, the capacity improvement is negligible, which thus shows the validness of the proposed distance criterion.}


{Fig.~\ref{fig_ratio_complexity} demonstrates the ratio of computation complexity between a benchmark and the proposed algorithm. The benchmark algorithm refers to the one where the objective function is directly calculated through the definition of permanent in (\ref{def_per}) without dividing the BS antenna array into subarrays. The other steps of the conventional scheme are the same as those of the proposed algorithm. From Fig.~\ref{fig_ratio_complexity}, we can find that the proposed algorithm can significantly reduce the computational complexity. This is because the calculation of the objective function can be simplified with the proposed sub-array-based method, as indicated in Remark~7.}
%

\begin{figure}[!tpb]
	\centering
	\vspace{0pt}
	\includegraphics[width=0.35\textwidth]{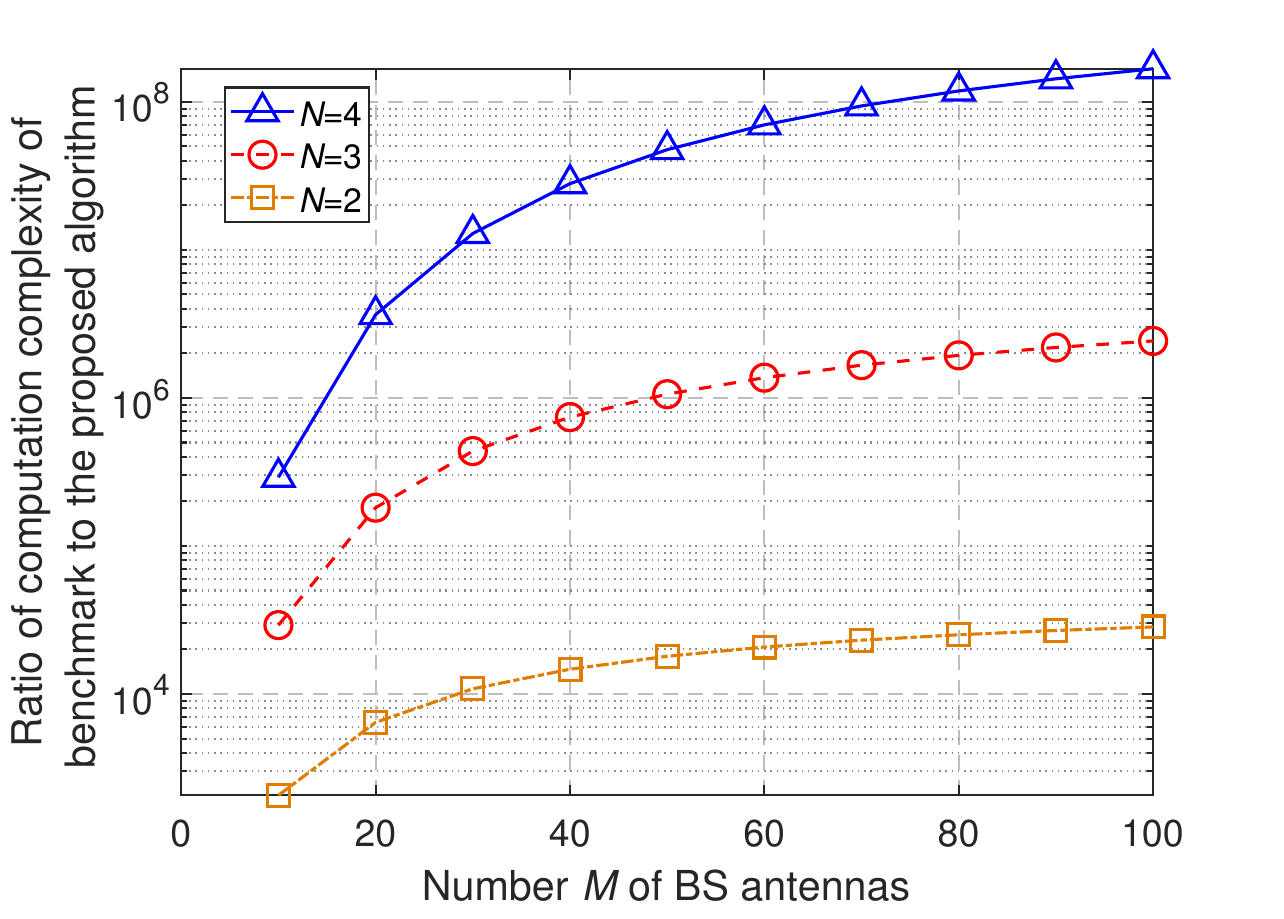}
	\caption{Ratio of computation complexity of a benchmark to the proposed algorithm with the number of antenna elements within each subarray $M_0=10$. The benchmark algorithm refers to the one where the objective function is directly calculated through the definition of permanent in (\ref{def_per}) without dividing the BS antenna array into subarrays. The other steps of the conventional scheme is the same as those of the proposed algorithm.}
	\vspace{-0.4cm}
	\label{fig_ratio_complexity}
\end{figure}

\begin{figure*}[!tpb]
	\centering
	\subfigure[Relationship between distance-dependent pathloss, i.e., $\beta_0d^{-\alpha}$, and \newline power allocation over BS antennas, given same XPD for all BS antennas]{
		\begin{minipage}[b]{0.49\textwidth}
			\centering
			\vspace{-0.2cm}
			\includegraphics[width=.7\textwidth]{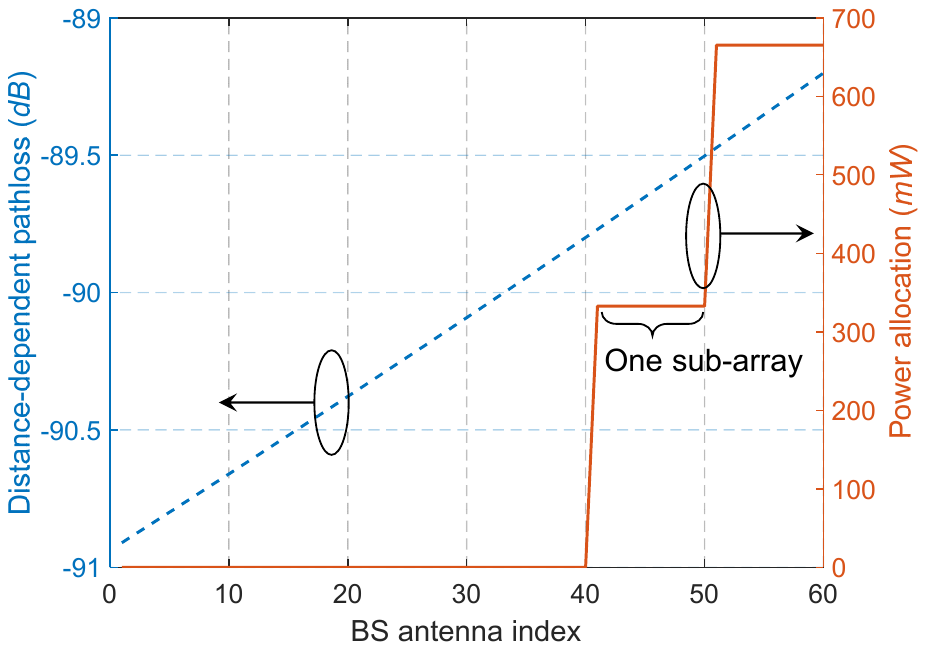}
			\label{fig_pathloss_vs_power}
	\end{minipage}}
	\subfigure[Relationship between XPD and power allocation over BS antennas, given the same pathloss for all BS antennas]{
		\begin{minipage}[b]{0.49\textwidth}
			\centering
			\vspace{-0.2cm}
			\includegraphics[width=.7\textwidth]{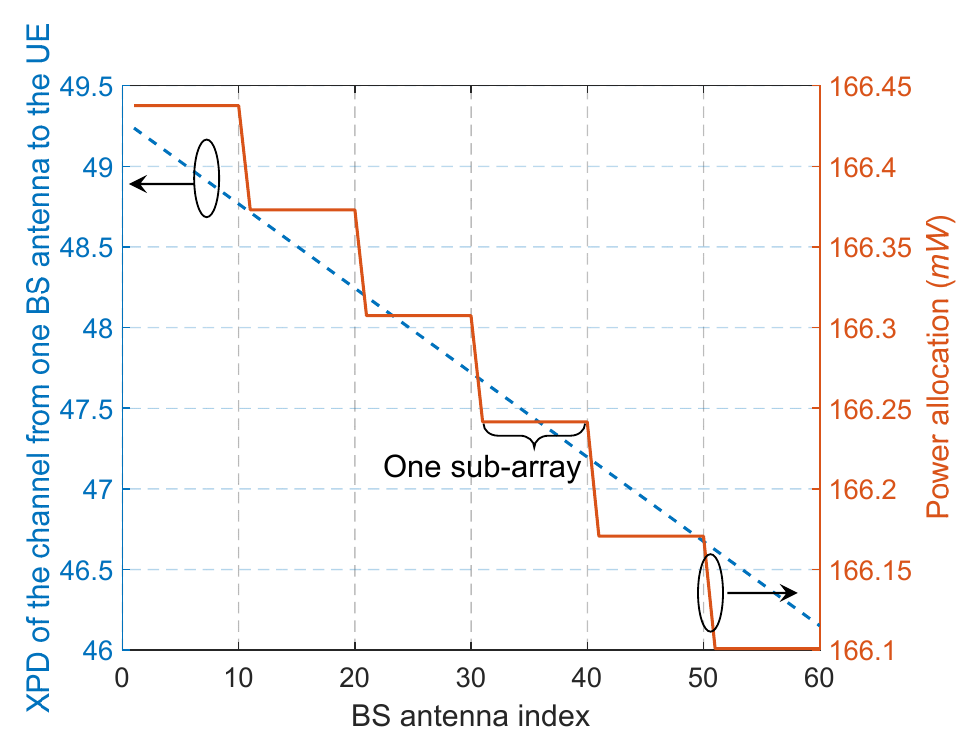}
			\label{fig_XPD_vs_power}
	\end{minipage}}
		\vspace{-0.1cm}
	\caption{Dependence of power allocations on XPD and pathloss over the BS antenna array, with the distance, azimuth angle, and elevation angle of the user given by $(r_U,\theta_U,\phi_U)=(30~m,\frac{\pi}{2},\frac{\pi}{2})$. A linear antenna array is deployed at BS, which can be divided into $S=6$ subarrays, each containing $M_0=10$ antenna elements.}
	\label{power_allocation}
	\vspace{-0.5cm}
\end{figure*}


Fig.~\ref{power_allocation} investigates the dependence of optimal power allocations, i.e., the diagonal elements of the optimal transmit covariance matrix acquired by the proposed algorithm, on XPD and distance-dependent pathloss. From Fig.~\ref{power_allocation}, we can find that both the variations of XPD and pathloss with the BS antennas can lead to unequal transmit power allocations, which is consistent with Theorem~\ref{theorem_conventional_not_apply}. Further, more transmit power should be allocated to antennas with higher XPD or less severe pathloss in order to maximize the system capacity. {Furthermore, Fig.~\ref{fig_pathloss_vs_power} shows that there exists large gap between the pathloss and the power allocation, i.e., most transmit power is allocated to a small portion of antennas while the antenna elements with small indices do not transmit signals. This is because the BS tends to allocate as much transmit power as possible to the antennas with less severe pathloss~\cite{Gao_eigenmode_2009}.}

\vspace{-0.2cm}
\section{Conclusion}
\label{sec_conclusion}
In this paper, we have investigated a downlink single-user XL-MIMO system using dual-polarized antennas. A dual-polarized XL-MIMO channel model has been proposed, where the variations of XPD and pathloss across BS antennas have been considered. Based on the proposed model, we have introduced a non-uniform XPD distance to complement existing near-far field boundary. Further, an efficient permanent-based transmit covariance matrix optimization algorithm has been proposed to maximize the ergodic capacity, which accounts for the non-uniform XPD and non-uniform pathloss. Based on our analysis and numerical results, the following conclusions are drawn:
\begin{itemize}
	\item When the BS-user distance is shorter than the derived non-uniform XPD distance or the aperture of the BS antenna array exceeds the derived non-uniform XPD aperture, propagation distances or AoDs vary across different BS antennas. Therefore, channels corresponding to different BS antennas exhibit distinct XPDs, which differs from the situation for conventional dual-polarized massive MIMO systems.
	\item Unlike classic Rayleigh distances, if non-uniform XPD distances are used for refined near-far field boundary, users with inclined directions are more likely to be located within the near-field region compared to those in the normal direction.
	\item Unlike conventional dual-polarized massive MIMO systems, the optimal transmit covariance matrix is not scalar matrix\footnote{A scalar matrix is a square matrix in which all of the principal diagonal elements are equal and the remaining elements are zero.} in dual-polarized XL-MIMO systems, where the transmit power (i.e., the diagonal elements of the transmit covariance matrix) tends to be allocated to BS antennas with less severe pathloss or higher XPD.
\end{itemize}

\begin{appendices}
	\vspace{-0.2cm}
	\section{Proof of Theorem~\ref{theorem_minimum_array_size}}
	\label{app_theorem_minimum_array_size}
	We first prove that when the aperture of the BS antenna array satisfies 
	\begin{align}
		\label{aperture}
		A>\frac{k}{1+k^2}(\frac{2r_U}{\eta\delta_U}(1-\frac{2}{\gamma^{th}_1+1})),
	\end{align}
	the difference of the XPD term $\chi_1$ between different BS antennas cannot be neglected, i.e.,
	\begin{align}
		\label{chi1_not_equal}
		\frac{\max_m\chi_1(d_{m})}{\min_m\chi_1(d_{m})}>\gamma^{th}_1.
	\end{align}
	\begin{remark}
		\label{remark_chi1_lower_bound}
		The ratio between the maximum value and the minimum value for $\chi_1$ over the diagonal antenna elements serves as a lower bound for the left-hand-side~(LHS) of (\ref{chi1_not_equal}), which represents the maximum value and the minimum value for $\chi_1$ over the whole antenna array, i.e.,
		\begin{align}
			\label{inequality_diagonal}
			\frac{\max_m\chi_1(d_{m})}{\min_m\chi_1(d_{m})}\ge\frac{\max_{m\in\mathcal{D}}\chi_1(d_{m})}{\min_{m\in\mathcal{D}}\chi_1(d_{m})}.
		\end{align}
		Here, $\mathcal{D}$ represents the set of the diagonal antenna elements.
	\end{remark}
	Therefore, we focus only on the diagonal antenna elements in the following to simplify analysis. 
	
	Note that the distance decay exponent in the expression for $\chi_1$, i.e., (\ref{influence_distance}), generally takes positive values~\cite{Shafi_polarized_2006}. Therefore, the maximum and minimum value of $\chi_1$ are achieved when the Tx-Rx distance is maximized and minimized, respectively. Note that the distance $r_U$ between the UE and the center of the BS antenna array is usually much larger than the size of the antenna array~\cite{Zeng_multi_user_RRS}. Therefore, the maximum and minimum Tx-Rx distances are achieved at the two end of the diagonal antenna elements. To this end, we try to determine the value of $\chi_1$ corresponding to these two antennas. 
	
	Consider the plane determined by the diagonal antenna elements and the UE. Further, define $\delta_U$ as the cosine value of the angle between the diagonal elements and the direction from the center of the BS antennas to the UE, whose expression can be found in (\ref{delta_U}). Denote the length of the diagonal of the antenna array by $D$. Then, the distance from the UE to the two antenna elements can be given by $r_U+\frac{D}{2}\delta_U$ and $r_U-\frac{D}{2}\delta_U$, respectively, where the corresponding values of the term $\chi_1$ can be given by $(r_U+\frac{D}{2}\delta_U)^\eta$ and $(r_U-\frac{D}{2}\delta_U)^\eta$, respectively. By applying absolute value operation, the maximum and minimum value of $\chi_1$ over the diagonal elements are
	\begin{align}
		\max_{m\in\mathcal{D}}\chi_1(d_{m})=(r_U+\frac{D}{2}|\delta_U|)^\eta,\\
		\min_{m\in\mathcal{D}}\chi_1(d_{m})=(r_U-\frac{D}{2}|\delta_U|)^\eta,
	\end{align}
	respectively. Therefore, according to Remark~\ref{remark_chi1_lower_bound}, we have
	\begin{align}
		\label{chi_1_lower_bound}
		\frac{\max_m\chi_1(d_{m})}{\min_m\chi_1(d_{m})}>\frac{(r_U+\frac{D}{2}|\delta_U|)^\eta}{(r_U-\frac{D}{2}|\delta_U|)^\eta}.
	\end{align}

	As we have mentioned, the BS-UE distance $r_U$ is much larger than the size of the antenna array. Therefore, the left-hand-side~(LHS) of (\ref{chi_1_lower_bound_thr}) can be simplified using a Taylor expansion, i.e.,
	\begin{align}
		\label{chi_1_lower_bound_approx}
		\frac{(r_U+\frac{D}{2}|\delta_U|)^\eta}{(r_U-\frac{D}{2}|\delta_U|)^\eta}\approx \frac{(r_U)^\eta(1+\frac{\eta D|\delta_U|}{2r_U})}{(r_U)^\eta(1-\frac{\eta D|\delta_U|}{2r_U})}.
	\end{align}
	Based on (\ref{chi_1_lower_bound_approx}), it is easy to find that when the aperture of the BS antenna satisfies (\ref{aperture}), we have
	\begin{align}
		\label{chi_1_lower_bound_thr}
		\frac{(r_U+\frac{D}{2}|\delta_U|)^\eta}{(r_U-\frac{D}{2}|\delta_U|)^\eta}>\gamma_1^{th}.
	\end{align}
	According to (\ref{chi_1_lower_bound}) and (\ref{chi_1_lower_bound_thr}), we can prove (\ref{chi1_not_equal}).
	
	We then show that when the aperture of the BS antenna array satisfies 
	\begin{align}
		\label{aperture_2}
		A>\frac{k}{1+k^2}(\frac{2\chi_2(\{\phi_0^{(l)}\}_l)(\gamma^{th}_2-1)}{(\gamma^{th}_2+1)b})^2,
	\end{align}
	the difference of the XPD term $\chi_2$ between different BS antennas cannot be neglected, i.e.,
	\begin{align}
		\label{chi2_not_equal}
		\frac{\max_m\chi_2(\{\phi_m^{(l)}\}_l)}{\min_m\chi_2(\{\phi_m^{(l)}\}_l)}>\gamma^{th}_2.
	\end{align}
	
	Similar to Remark~\ref{remark_chi1_lower_bound}, the ratio between the maximum value and the minimum value for $\chi_2$ over the diagonal antenna elements serves as a lower bound for the left-hand-side~(LHS) of (\ref{chi2_not_equal}), i.e.,
	\begin{align}
		\label{chi2_lower_bound}
		\frac{\max_m\chi_2(\{\phi_m^{(l)}\}_l)}{\min_m\chi_2(\{\phi_m^{(l)}\}_l)}>\frac{\max_{m\in\mathcal{D}}\chi_2(\{\phi_m^{(l)}\}_l)}{\min_{m\in\mathcal{D}}\chi_2(\{\phi_m^{(l)}\}_l)}.
	\end{align}
	Therefore, we only focus on the diagonal antenna elements in the following to simplify analysis. 
	
	Denote the difference between the AoD of the $m$-th antenna element and that of the central antenna element as $\Delta\phi_m^{(l)}$, i.e., $\Delta\phi_m^{(l)}\triangleq\phi_m^{(l)}-\phi_0^{(l)}$.
	Since the distance $r_U$ between the BS and the UE is much larger than the aperture size of the BS antenna array, $\Delta\phi_m^{(l)}$ takes a small value. To this end, $\chi_2(\{\phi_m^{(l)}\}_l)$ can be approximated by its first-order Taylor expansion, i.e.,
	\begin{align}
		\label{chi2_Taylor_approximation}
		\chi_2(\{\phi_m^{(l)}\}_l)\approx \chi_2(\{\phi_0^{(l)}\}_l)+\sum_l \left.\dfrac{\partial \chi_2 }{\partial \phi_m^{(l)}}\right |_{\{\phi_m^{(l)}\}_l=\{\phi_0^{(l)}\}_l}\Delta\phi_m^{(l)}.
	\end{align} 
	Further, the difference $\Delta\phi_m^{(l)}$ between AoDs is linear with respect to antenna index $m$, i.e.,
	\begin{align}
		\label{linear_difference_AoD}
		\Delta\phi_m^{(l)}=\frac{-\delta^{(l)}\sqrt{1-(\delta^{(l)})^2}}{|\delta^{(l)}|r^{(l)}_0\sin\theta^{(l)}_0}md_D.
	\end{align}
	Here, $d_D$ represents the spacing between two adjacent diagonal antenna elements, $\delta^{(l)}=\frac{\cos\phi_0^{(l)}+k\sin\phi_0^{(l)}}{\sqrt{1+k^2}}$ represents the cosine value of the angle between the diagonal direction and the direction from the center of the BS antenna array to the $l$-th cluster, and $r^{(l)}_0\sin\theta^{(l)}_0$ corresponds to the distance between the projection of the $l$-th cluster on the $x-y$ plane and the center of the BS antenna array. Based on (\ref{chi2_Taylor_approximation}) and (\ref{linear_difference_AoD}), we can find that the maximum and minimum values of $\chi_2(\{\phi_m^{(l)}\}_l)$ can be achieved at the two end of the diagonal antenna elements, which can be expressed as
	\begin{align}
		\max_{m\in\mathcal{D}}\chi_2(\{\phi_m^{(l)}\}_l)=\chi_2(\{\phi_0^{(l)}\}_l)+b\frac{D}{2},\\
		\label{min_chi2}
		\min_{m\in\mathcal{D}}\chi_2(\{\phi_m^{(l)}\}_l)=\chi_2(\{\phi_0^{(l)}\}_l)-b\frac{D}{2},
	\end{align}
	respectively. Here, $D$ is the length of the diagonal of the antenna array, and the definition of $b$ can be found in (\ref{expression_b}). Therefore, according to (\ref{chi2_lower_bound}), we have
	\begin{align}
		\label{chi2_lower_bound_v2}
		\frac{\max_m\chi_2(\{\phi_m^{(l)}\}_l)}{\min_m\chi_2(\{\phi_m^{(l)}\}_l)}>\frac{\chi_2(\{\phi_0^{(l)}\}_l)+b\frac{D}{2}}{\chi_2(\{\phi_0^{(l)}\}_l)-b\frac{D}{2}}.
	\end{align}
	Based on (\ref{chi2_lower_bound}), it is easy to find that (\ref{chi2_not_equal}) holds when the aperture of BS antenna array satisfies (\ref{aperture_2}), which ends proof.
	\vspace{-0.4cm}
	\section{Proof of Theorem~\ref{theorem_maximum_user_distance}}
	\label{app_theorem_maximum_distance}
	The main idea for proving Theorem~\ref{theorem_maximum_user_distance} is similar to that for proving Theorem~\ref{theorem_minimum_array_size}. Following the proof procedure from (\ref{inequality_diagonal}) to (\ref{chi_1_lower_bound_approx}), we can arrive at
	\begin{align}
		\label{lower_bound_for_ratio}
		\frac{\max_m\chi_1(d_{m})}{\min_m\chi_1(d_{m})}>\frac{1+\frac{\eta D|\delta_U|}{2r_U}}{1-\frac{\eta D|\delta_U|}{2r_U}}.
	\end{align}
	
	Based on (\ref{lower_bound_for_ratio}), when the distance $r_U$ between the user and the center of the BS antenna array satisfies
	\begin{align}
		\label{r_U_th_normal}
		r_U\le \frac{\gamma_1^{th}+1}{\gamma_1^{th}-1}\frac{\eta D \delta_U}{2},
	\end{align}
	we have $\frac{\max_m\chi_1(d_{m})}{\min_m\chi_1(d_{m})}>\gamma_1^{th}$. 
	
	{
	We should note that the above derivation is based on the assumption that the maximum and minimum propagation distances are achieved at the two ends of diagonal elements, which, however, does not hold when the near-field region corresponds to a very short distance. In particular, a situation where this does hold occurs when
	\begin{align}
		\label{user_dir}
		\delta_U\le \sqrt{1-\frac{4}{(\gamma_1^{th})^{\frac{2}{\eta}}+3}}.
	\end{align}
	Therefore, to make the derivation more rigorous, we will deal with this case in the following. Specifically, when the user direction satisfies the condition in (\ref{user_dir}), the minimum Tx-Rx distance is achieved by the antenna element that lies in the projection point of the user over the diagonal of the antenna array, while the maximum Tx-Rx distance is achieved by the antenna element at one end of the diagonal antenna elements. Therefore, we have
	\begin{align}
		\label{max_counter}
		\max_{m\in\mathcal{D}}\chi_1(d_{m})=(\sqrt{r_U^2+\frac{D^2}{4}+\delta_Ur_UD})^\eta,\\
		\label{min_counter}
		\min_{m\in\mathcal{D}}\chi_1(d_{m})=(r_U\sqrt{1-\delta_U^2})^\eta.
	\end{align}
	Based on (\ref{max_counter}) and (\ref{min_counter}), we can find that when the distance $r_U$ between the user and the center of the BS antenna array satisfies
	\begin{align}
		\label{r_U_th_counter}
		r_U\le\frac{-D\delta_U-\sqrt{((\gamma_1^{th})^\frac{2}{\eta}-1)D^2(1-\delta_U^2)}}{2(1-(\gamma_1^{th})^\frac{2}{\eta}+(\gamma_1^{th})^\frac{2}{\eta}\delta_U^2)},
	\end{align}
	we have $\frac{\max_m\chi_1(d_{m})}{\min_m\chi_1(d_{m})}>\gamma_1^{th}$. By combining (\ref{r_U_th_normal}), (\ref{user_dir}), and (\ref{r_U_th_counter}), we can derive the contribution of the XPD term $\chi_1$ on the non-uniform XPD distance denoted by $r_1^{th}$, whose expression is given in (\ref{pattern_feed}).}
	
	{
	Then, we move on to discuss the AoD-dependent XPD term $\chi_2$, and prove that when the propagation distance $r_U$ is shorter than the non-uniform XPD distance in (\ref{expression_distance_th}), we have
	\begin{align}
		\label{chi2_not_equal_v2}
		\frac{\max_m\chi_2(\{\phi_m^{(l)}\}_l)}{\min_m\chi_2(\{\phi_m^{(l)}\}_l)}>\gamma^{th}_2.
	\end{align}
	
	As mentioned in Appendix~\ref{app_theorem_minimum_array_size}, the ratio between the maximum value and the minimum value for $\chi_2$ over the diagonal antenna elements serves as a lower bound for the left-hand-side~(LHS) of (\ref{chi2_not_equal_v2}), i.e.,
	\begin{align}
		\label{chi2_lower_bound_v3}
		\frac{\max_m\chi_2(\{\phi_m^{(l)}\}_l)}{\min_m\chi_2(\{\phi_m^{(l)}\}_l)}>\frac{\max_{m\in\mathcal{D}}\chi_2(\{\phi_m^{(l)}\}_l)}{\min_{m\in\mathcal{D}}\chi_2(\{\phi_m^{(l)}\}_l)}.
	\end{align}
	Therefore, we focus only on the diagonal antenna elements in the following to simplify the analysis.
	
	Following (\ref{chi2_Taylor_approximation})- (\ref{min_chi2}), we can find the maximum value and minimum value of $\chi_2$ over the diagonal elements of the antenna array, as shown below:
	\begin{align}
		\label{max_chi_2}
		&\max_{m\in\mathcal{D}}\chi_2(\{\phi_m^{(l)}\}_l)\notag\\
		&=\chi_2(\{\phi_0^{(l)}\}_l)+\left|\sum_{l}\frac{\partial \chi_2}{\partial \phi_m^{(l)}} \bigg|_{\phi_m^{(l)}=\phi_0^{(l)}}\frac{\sin(\arctan(k)+\phi_0^{(l)})D}{2\sin\theta_0^{(l)}r_0^{(l)}}\right|,\\
		\label{min_chi_2}
		&\min_{m\in\mathcal{D}}\chi_2(\{\phi_m^{(l)}\}_l)\notag\\
		&=\chi_2(\{\phi_0^{(l)}\}_l)-\left|\sum_{l}\frac{\partial \chi_2}{\partial \phi_m^{(l)}} \bigg|_{\phi_m^{(l)}=\phi_0^{(l)}}\frac{\sin(\arctan(k)+\phi_0^{(l)})D}{2\sin\theta_0^{(l)}r_0^{(l)}}\right|.
	\end{align}
	Note that as the user moves away from the BS, the distance between the BS and scatterer clusters also becomes larger~\cite{Zhou_geo_2019}. To characterize this trend, we model the distance $r_0^{(l)}$ between the $l$-th cluster and the BS by
	\begin{align}
		\label{distance_scatter}
		r_0^{(l)}=c^{(l)}r_U.
	\end{align}
	Based on (\ref{max_chi_2}), (\ref{min_chi_2}), and (\ref{distance_scatter}), we can find that (\ref{chi2_not_equal_v2}) holds when the distance between the BS and the UE is shorter than the non-uniform XPD distance, which ends the proof.}


	\section{Proof of Lemma~\ref{lemma_convetional_MIMO_identity}}
	\label{app_lemma_convetional_MIMO_identity}
	According to~\cite{Gao_eigenmode_2009}, the optimal transmit covariance matrix is diagonal when MIMO channels satisfy
	\begin{align}
		\label{MIMO_channel_form_for_diagonal}
		\bm{G}=\bm{D}+\bm{M}\odot\bm{G_{iid}},
	\end{align}
	where $\bm{D}$ is a $2N\times 2M$ deterministic matrix with at most one nonzero element in each row and each column, $\bm{M}$ is a $2N\times 2M$ deterministic matrix with nonnegative elements, and $\bm{G}_{iid}$ is a $2N\times 2M$ random matrix with independent and identically distributed~(i.i.d.) elements having zero means. Therefore, by letting $\bm{D}=\bm{0}$, and setting $\bm{M}$ and $\bm{G_{iid}}$ as the matrix composed of pathloss and small-scale fading, respectively, the conventional dual-polarized massive MIMO channel can be written in the form of (\ref{MIMO_channel_form_for_diagonal}). Consequently, the optimal transmit covariance is a diagonal matrix for conventional massive MIMO systems with dual-polarized antennas, where the diagonal elements correspond to the transmit power of different BS antennas.
	
	Note that in the conventional massive MIMO with dual-polarized antennas, the antenna array is much smaller than that in dual-polarized XL-MIMO systems. Therefore, distance-dependent pathloss and XPD for different BS antennas can be regarded as the same. Therefore, based on~(\ref{pathloss}), we have
	\begin{align}
		\label{same_pathloss_copolarized}
		\beta^{(HH)}_{m_1}=\beta^{(HH)}_{m_2}=\beta^{(VV)}_{m_1}=\beta^{(VV)}_{m_2},\\
		\label{same_pathloss_crosspolarized}
		\beta^{(VH)}_{m_1}=\beta^{(VH)}_{m_2}=\beta^{(HV)}_{m_1}=\beta^{(HV)}_{m_2},
	\end{align}
	for any BS antennas $m_1$ and $m_2$. Note that the small-scale fading has the same distribution among different BS antennas and different polarizations. Therefore, based on (\ref{same_pathloss_copolarized}), we can derive that the co-polarized channel components corresponding to different BS antennas and polarizations, i.e., $\{g_{n,1}^{(VV)},\dots,g_{n,M}^{(VV)},g_{n,1}^{(HH)},\dots,g_{n,M}^{(HH)}\}$, are identically distributed. Similarly, based on (\ref{same_pathloss_crosspolarized}), it can be found that the cross-polarized channel components have the same distribution. Therefore, we can conclude that the transmit power should be equally allocated among different antennas and polarizations, indicating that the optimal transmit covariance matrix is a scalar matrix, i.e., 
	\begin{align}
		\bm{Q}^*=\frac{1}{2M}\bm{I}_{2M}.
	\end{align}
	\vspace{-.2cm}
	\section{Proof of Remark~\ref{remark_complexity_bound}}
	\label{app_remark_complexity_bound}
	Note that both $(\frac{(k-1)2N!}{(2N-k)!}+1)$ and $(\frac{(2S+k)!}{(k+1)!(2S)!})$ increase with $k$. Therefore, they can be upper bounded by
	\begin{align}
		\label{upper_bound_1}
		\frac{(k-1)2N!}{(2N-k)!}+1\le (2N-1)2N!+1,
	\end{align}
	\begin{align}
		\label{upper_bound_2}
		\frac{(2S+k)!}{(k+1)!(2S)!}\le \frac{(2S+2N)!}{(2N+1)!(2S)!},
	\end{align}
	respectively. Further, since $(2N-1)2N!\gg 1$, the right-hand-side of (\ref{upper_bound_1}) can be approximated by
	\begin{align}
		\label{approx}
		(2N-1)2N!+1\approx (2N-1)2N!.
	\end{align}
	By substituting (\ref{upper_bound_1}), (\ref{upper_bound_2}), and (\ref{approx}) into the expression of $\mathbb{N}_C^{sim}$ in (\ref{computation_complexity_simplify}), we can arrive at (\ref{upper_bound_computation_complexity}), which ends the proof.

\end{appendices}

\end{document}